\documentclass[lettersize,journal]{IEEEtran}
\usepackage{amsmath,amsfonts}
\usepackage{amsthm}
\usepackage[ruled,vlined]{algorithm2e}
\usepackage{mdframed} % algorithm beautify
\usepackage{array}
\usepackage[caption=false]{subfig}
\usepackage{subfloat}
\usepackage{textcomp}
\usepackage{stfloats}
\usepackage{url}
\usepackage{verbatim}
\usepackage{mathrsfs}
\usepackage{pifont}
\usepackage{booktabs}
\usepackage{threeparttable} % 三线表
\usepackage[table]{xcolor} % 导入 xcolor 包以支持表格着色
\usepackage{multirow} % 支持多行合并
\usepackage{multicol}
\usepackage{makecell} % 单元格内换行
\usepackage{graphicx} % 插入图片
\usepackage[
    n,  % or lambda 
    advantage,
    operators,
    sets,
    adversary,
    landau,
    probability,
    notions,
    logic,
    ff, 
    mm,
    primitives,
    events,
    complexity,
    oracles,
    asymptotics,
    keys
]{cryptocode}
\usepackage{stmaryrd}
\createprocedureblock{procb}{, boxed}{}{}{}
\newcommand{\md}{\!\!\!\!\mod}
\newcommand{\lb}{\llbracket}
\newcommand{\rb}{\rrbracket}
\newcommand{\pdec}{$\mathsf{PDec}$\ }
\newcommand{\tdec}{$\mathsf{TDec}$\ }
\newtheorem{definition}{Definition}
\newtheorem{theorem}{Theorem}
\newtheorem{lemma}{Lemma}
\definecolor{myblue}{HTML}{0000CC}
\usepackage[colorlinks=true]{hyperref}
\hypersetup{
    linkcolor=myblue,
    urlcolor=myblue,
    citecolor=myblue
}

\definecolor{reecolor}{RGB}{245, 245, 245} % CP颜色 (诚实但好奇)
\definecolor{teecolor}{RGB}{215, 215, 215} % TEE颜色 (完全可信)
\definecolor{whitecolor}{RGB}{253, 253, 253} % 白色 (非常浅的灰色)

\begin{document}

\title{TRUST: A Toolkit for TEE-Assisted Secure Outsourced Computation over Integers}

\author{Bowen Zhao, ~\IEEEmembership{Member,~IEEE},
        Jiuhui Li, 
        Peiming Xu$^*$, 
        Xiaoguo Li,
        Qingqi Pei, ~\IEEEmembership{Senior Member,~IEEE}, \\
        and Yulong Shen, ~\IEEEmembership{Member,~IEEE}
\thanks{Bowen zhao is with Guangzhou Institute of Technology, Xidian University, Guangzhou 510555, China and also with the State Key Laboratory of Public Big Data, Guizhou University, Guiyang 550025, China (e-mail: bwinzhao@gmail.com).

Jiuhui Li is with Guangzhou Institute of Technology, Xidian University, Guangzhou 510555, China (e-mail: lijiuhui@stu.xidian.edu.cn).

Peiming Xu is with Electric Power Research Institute, CSG, Guangzhou Guangdong, 510663, China, and also with Guangdong Provincial Key Laboratory of Power System Network Security, Guangzhou Guangdong, 510663, China (email: xupm@csg.cn).

Xiaoguo Li is with the College of Computer Science, Chongqing University, Chongqing 400044, China (e-mail: csxgli@cqu.edu.cn).

Qingqi Pei and Yulong Shen are with the State Key Laboratory of Integrated Service Networks, Xidian University, Xi’an 710126, China (e-mail: qqpei@mail.xidian.edu.cn, ylshen@mail.xidian.edu.cn).

(Corresponding author: Peiming Xu)}
}

% The pars
\markboth{Journal of \LaTeX\ Class Files,~Vol.~~, No.~~, ~~202$\times$}%
{Shell \MakeLowercase{\textit{et al.}}: A Sample Article Using IEEEtran.cls for IEEE Journals}

% \IEEEpubid{0000--0000/00\$00.00~\copyright~2021 IEEE}
% Remember, if you use this you must call \IEEEpubidadjcol in the second
% column for its text to clear the IEEEpubid mark.

\maketitle

\begin{abstract}
Secure outsourced computation (SOC) provides secure computing services by taking advantage of the computation power of cloud computing and the technology of privacy computing (e.g., homomorphic encryption). Expanding computational operations on encrypted data (e.g., enabling complex calculations directly over ciphertexts) and broadening the applicability of SOC across diverse use cases remain critical yet challenging research topics in the field. Nevertheless, previous SOC solutions frequently lack the computational efficiency and adaptability required to fully meet evolving demands. To this end, in this paper, we propose a toolkit for TEE-assisted (Trusted Execution Environment) SOC over integers, named TRUST. In terms of system architecture, TRUST falls in a single TEE-equipped cloud server only through seamlessly integrating the computation of REE (Rich Execution Environment) and TEE. In consideration of TEE being difficult to permanently store data and being vulnerable to attacks, we introduce a (2, 2)-threshold homomorphic cryptosystem to fit the hybrid computation between REE and TEE. Additionally, we carefully design a suite of SOC protocols supporting unary, binary and ternary operations. To achieve applications, we present \texttt{SEAT}, secure data trading based on TRUST. Security analysis demonstrates that TRUST enables SOC, avoids collusion attacks among multiple cloud servers, and mitigates potential secret leakage risks within TEE (e.g., from side-channel attacks). Experimental evaluations indicate that TRUST outperforms the state-of-the-art and requires no alignment of data as well as any network communications. Furthermore, \texttt{SEAT} is as effective as the \texttt{Baseline} without any data protection.
\end{abstract}


\section{Introduction}
\IEEEPARstart{O}{utsourcing} computation allows a user with limited resources to outsource operations (or rather, computations) over data to a cloud server. Secure outsourced computation (SOC) ensures computations over data without disclosing the data itself \cite{shan2018practical}. In other words, SOC provides not only convenient computing services for the user with limited resources, but also privacy protection for the outsourced data. Broadly speaking, many types of secure operations can be considered as SOC, for instance, searchable encryption \cite{bosch2014survey}, private set intersection \cite{chen2017fast}, secure training or inference for machine learning \cite{rathee2020cryptflow2}, and secure computation as a service \cite{zhao2024pega}. Thus, SOC can be applied to various fields including search, product recommendation, artificial intelligence, production optimization, etc.

Technically, a cloud server providing SOC services usually adopts privacy computing technologies to safeguard the privacy of outsourced data. Homomorphic encryption (HE), secure multi-party computation (MPC), and trusted execution environment (TEE) \cite{platform2013global} are common privacy computing technologies. HE allows operations over encrypted data (also known as ciphertexts) directly and accesses no decryption key \cite{zhao2023soci}. HE includes partially HE achieving addition or multiplication only and fully HE supporting addition and multiplication simultaneously. MPC enables multiple parties jointly performing an operation and safeguards each party's inputs \cite{lindell2020secure}. SOC solutions based on multiple servers usually rely on MPC. Either HE or MPC achieves SOC through software method, while TEE is through a hardware method. TEE creates an enclave offering hardware-based memory encryption for code and data, in which no unauthorized entities from outside the enclave access inside code or data \cite{li2023survey}. TEE offers a promising hardware-based approach to secure computation, achieving CPU-level performance and effectively bridging the gap between academic research and industrial adoption in the field.

Existing SOC solutions remain constrained by limited computational operations and often assume no collusion among cloud servers. On the one hand, SOC based on HE usually supports secure addition, secure multiplication, or both. However, in practice, a user requires more operations (e.g., secure comparison) offered by SOC to handle outsourced data. On the other hand, although SOC based on MPC offers more secure operations, it requires at least two cloud servers and assumes that there is no collusion between cloud servers. In this case, more cloud servers mean more deployment costs. Additionally, more network communications and extra data alignment operations are required to enable secure computations. Furthermore, arguably, the assumption of no collusion may be too strong, as it is difficult to determine whether cloud servers launch the collusion attack or not. SOC based on TEE enriches computational operations and avoids the collusion attack launched by multiple cloud servers \cite{liu2020lightning}. Unfortunately, any enclave holds limited storage space, and its storage is random access memory (RAM) \cite{sabt2015trusted}. Thus, it is not trivial to permanently store and process outsourced data in an enclave. Furthermore, TEE is not always trusted because side-channel attacks are high likely to expose inside code and data in the enclave \cite{wu2022hybrid}.

In order to mitigate the above limitations faced by SOC, in this work, we concentrate on devising a hybrid SOC framework by integrating advantages from the software-based method and the hardware-based method. Intuitively, this framework is confronted with several challenges. The \textbf{first} one is to balance the software method and the hardware method, and then eliminate collusion attacks. The software-based method and the hardware-based method follow entirely different designs achieving SOC. A basic operation is supported by the former but not by the latter, and vice versa. The \textbf{second} one is to fulfill secure computations under lacking a fully trusted enclave. Due to the degradation of trust in enclave, some excellent features provided by fully trusted enclave cannot be adopted. Thus, this significantly limits the functionality of the enclave. The \textbf{last} one is to enrich secure operations falling this framework. To eliminate the collusion attack, this framework should not utilize MPC to enable SOC, which restricts secure operations. Moreover, the enclave in this framework fails to enrich secure operations by reason of trust concerns. Worse, HE-based SOC only achieves simple secure operations.

In response to the aforementioned challenges, in this work, we propose TRUST\footnote{TRUST: a toolkit for \textbf{T}EE-assistant secu\textbf{R}e o\textbf{U}t\textbf{S}ourced compu\textbf{T}ation over integer}, a toolkit for TEE-assisted SOC over integer. In brief, the proposed TRUST employs a single TEE-equipped cloud server, in which a rich execution environment (REE) of the cloud server cooperates with TEE to provide SOC services. A (2, 2)-threshold Paillier cryptosystem \cite{zhao2024soci} is introduced in TRUST to construct SOC protocols. After that, we come up with secure data trading based on TRUST, called \texttt{SEAT}\footnote{\texttt{SEAT}: \texttt{SE}cure data tr\texttt{A}ding based on \texttt{T}RUST} to explore the application of SOC. Our innovations in this work involve three fronts.
\begin{itemize}
\item \textbf{Practical SOC framework}. We propose a practical SOC framework that relies on a single TEE host, where the TEE host calls a software-based method (i.e., the threshold Paillier cryptosystem \cite{zhao2024soci}), and the enclave of TEE offers hardware-based capabilities. The TEE host and the enclave serve as the two parties of SOC and jointly perform SOC protocols. Particularly, the proposed SOC framework does not require a fully trusted enclave.
\item \textbf{Rich SOC operations}. Following the proposed framework and assumption, we carefully design a suite of underlying SOC protocols encompassing unary, binary and ternary operations.
\item \textbf{Feasible SOC application}. We apply the proposed TRUST to data trading and present a solution \texttt{SEAT}, which not only fulfills the outsourced task requirements but also effectively avoids data resale and privacy leakage.

\end{itemize}

The rest of this work is organized as follows. We briefly review related work in Section \ref{Related Work}. Subsequently, Section \ref{Preliminarys} presents the preliminaries essential for constructing TRUST. After that, we present the system model and threat model of TRUST and elaborate on its design in Section \ref{System Model and Threat Model} and Section \ref{TRUST Design}, respectively. In Section \ref{SEAT Details} and Section \ref{Security Analysis}, the application and the security of TRUST are given, respectively. We extensively evaluate the efficiency and effectiveness of both TRUST and \texttt{SEAT} in Section \ref{Experimental Evaluation}. Finally, we conclude this work in Section \ref{Conclusion}.

\section{Related Work}\label{Related Work}
The methods for achieving SOC can be broadly classified into the software-based (e.g., HE, MPC) one and the hardware-based (e.g., TEE) one. In this section, we briefly review these two approaches.

\subsection{Software-based SOC}
Recently, significant advancements in cryptographic primitives and protocols have expanded secure computation capabilities, enabling diverse operations on ciphertexts (e.g., arithmetic, logical) and supporting various numeric representations (e.g., integers, floating-point numbers) \cite{shan2018practical} \cite{yang2019comprehensive}. Despite continuous optimization and integration of technologies to enhance privacy and performance, achieving a balance between security and efficiency remains a fundamental challenge in privacy-preserving outsourced computation.

% work \cite{elmehdwi2014secure} 缺点 sk完全给了服务器，然后是用随机数掩码来处理的
% 这一段要写单架构下的HE方案，限制了非线性操作 
HE-based solutions (e.g., \cite{erkin2012generating} \cite{samanthula2013efficient} \cite{elmehdwi2014secure} \cite{veugen2015secure}) enable direct computation on ciphertexts. Erkin \textit{et al}. \cite{erkin2012generating} implemented personalized recommendation generation over ciphertexts by jointly using two additively homomorphic cryptosystems, including the Paillier cryptosystem \cite{paillier1999public} and Damgård-Geisler-Krøigaard (DGK) \cite{damgaard2007efficient}. The work \cite{elmehdwi2014secure} proposed a novel S$k$NN protocol, which is also based on the Paillier cryptosystem to enable $k$-nearest neighbor search over ciphertexts in the cloud, providing key protocols such as secure multiplication, etc., to accomplish this goal. Although the Paillier cryptosystem has the capability to extend operation types over ciphertexts, some of solutions (e.g., \cite{elmehdwi2014secure}, \cite{veugen2015secure}) suffer from stability issues, as they rely on entrusting the complete decryption key to a single decryption server, which introduces a critical single point of failure. While, other solutions (e.g., somewhat HE \cite{damgaard2012multiparty} and fully HE \cite{gentry2009fully}) either lack support for arbitrary computations or remain impractical due to substantial computational overhead.

% 双架构下的HE方案，比较强烈的非共谋假设以及开销问题
The integration of a twin-server architecture with HE has emerged as a novel paradigm, supporting arbitrary operations while mitigating the single point of security failure. The schemes in \cite{zhao2023soci}, \cite{zhao2024soci}, \cite{schoenmakers2015universally}, \cite{liu2016privacy} adopted the Paillier cryptosystem with threshold decryption \cite{lysyanskaya2001adaptive} within a twin-server architecture to enrich computational operations over ciphertexts. Liu \textit{et al}. \cite{liu2016privacy} proposed a framework named POCF, which enhances operations over ciphertexts and facilitates collaborative decryption by distributing the Paillier cryptosystem private key between two servers, ensuring that neither can independently access the original data. Importantly, this design enables collaborative decryption by the twin-server, allowing the cloud to complete the SOC task and return the result to the client in a single round of communication, while preserving privacy. Likewise, the work \cite{zhao2023soci} implemented a toolkit for SOC on integers, named SOCI, which falls in a twin-server architecture utilizing the threshold Paillier cryptosystem. In particular, the toolkit provides more efficient secure multiplication and secure comparison protocols compared to those proposed by \cite{liu2016privacy}. Furthermore, SOCI\textsuperscript{+}, as proposed by the work \cite{zhao2024soci}, significantly outperforms SOCI \cite{zhao2023soci}, attributed to a more efficient variant of the threshold Paillier cryptosystem and an offline-online mechanism to enhance precomputation capabilities. However, all the aforementioned work fundamentally relies on the non-collusion assumption between the twin-server, which falls short of real-world requirements. Even multi-server applications using a (t, n)-threshold Paillier cryptosystem still require the servers to refrain from collusion.

% GC和SS一起写，更庞大的计算和通讯轮次开销，客户端需要参与的多轮次交互形式
MPC-based solutions, which are more inclined to extend SOC for specific applications, are represented by Yao’s garbled circuits (GC) \cite{yao1986generate} and secret sharing (SS) \cite{shamir1979share}, and enable multiple parties to jointly compute a result without revealing their private inputs. Blanton \textit{et al}. \cite{blanton2012secure_sequence_comparisons} presented a secure sequence comparison scheme (i.e., edit distance) that employs GC in a novel non-black-box manner with two servers, suitable for privacy-preserving DNA alignment and other string matching tasks. A scheme for iris identification is proposed in \cite{blanton2012secure_iris_matching}, which enables encrypted iris code matching on untrusted servers while achieving privacy via predicate encryption in a single-server model or secure MPC with SS in a multi-server model. Additionally, ABY \cite{demmler2015aby} and ABY 2.0 \cite{patra2021aby2} employed SS to implement mixed frameworks for secure two-party computation, integrating arithmetic and boolean. Although these schemes support general computation for arbitrarily complex functions, they demand substantial storage and communication resources and always rely on the non-collusion assumption among parties. To avoid existing limitations, the design of TRUST aims to effectively balance computational efficiency and system generality, achieving an optimal synergy between the two.

\subsection{Hardware-based SOC}
%总起介绍目前TEE的研究侧重于什么
Benefiting from the ability of TEE to prevent malicious software attacks and ensure privacy protection, along with strong guarantees of data confidentiality and integrity, TEE-based secure computations are emerging \cite{li2023survey}, \cite{zheng2021survey}.

% 开始举例，先提方案广阔性，再具体到SOC
Building on the above core capabilities, existing TEEs (e.g., Intel\textsuperscript{\textregistered} Software Guard Extensions (SGX) \cite{costan2016intel}, ARM's TrustZone) underpin advancements of secure computing research and trusted computing applications. Representative examples encompass privacy-preserving outsourced computation \cite{liu2020lightning}, secure credential migration \cite{wang2023ma}, secure root-of-trust design \cite{hoang2022trusted}, secure pipeline management \cite{gao2023enabling}, formal verification framework \cite{sun2020design}, secure electronic healthcare data sharing \cite{li2024trusthealth}, robust federated learning \cite{cao2024srfl}, identity authentication and authorization services \cite{li2024make}, among others.

% 然而这些方案只考虑了TEE在保护数据上的强大能力， 
% 优点：适合特定的高安全性需求场景
% 缺点①：计算资源限制  
% 缺点②：特定功能支持 不支持像通用 CPU 那样的全面计算功能
% 缺点③：运行环境限制 无法运行所有类型的应用程序或复杂的操作系统功能
% 缺点④：扩展性限制（需要依赖外部的计算资源，不能完全依赖于 TEE 内部）

% TEE的多实体应用的例子
Liu \textit{et al}. \cite{liu2022privacy} presented a SOC scheme integrating blockchain, smart contracts, a TEE-equipped cloud service provider, and a trusted authority responsible for managing cryptographic keys as a fully trusted third party. The scheme involves multi-round data transmission among entities, increasing communication overhead, while the deployment of extensive processing tasks on-chain introduces notable computational costs.

% 开销高的一个例子
The previous works \cite{liu2020lightning} combined SS with the threshold Paillier cryptosystem to design a fast processing toolkit on a single cloud server, aiming to resist side-channel attacks against TEE. Yet, this approach relies on multiple trusted processing units, leading to relatively high computation costs and frequent inter-server communication rounds during protocol execution. The work \cite{dai2019sdte} put forward SDTE, a blockchain-based secure data trading system utilizing TEE and AES-256 symmetric encryption. However, it relies on TEE to fully decrypt data using symmetric keys and directly operates over plaintext, neglecting the vulnerability of TEE to side-channel attacks. More than the aforementioned work, many existing TEE-based solutions leverage its strong data protection capabilities, but they do not fully consider the limitations of TEE (e.g., restricted computational resources and scalability constraints), which hinder its applicability in complex, resource-intensive scenarios. TEE is not bullet-proof and has been successfully attacked numerous times in various ways, as shown in \cite{xu2015controlled}, \cite{chen2019sgxpectre}, \cite{liu2015last}. Undoubtedly, treating TEE as a trusted third party to directly perform operations on plaintext or to maintain a master decryption key fails to adequately consider the vulnerability to side-channel attacks, which could lead to privacy leakage and undermine its trustworthiness. 

Inspired by the design proposed in the literatures \cite{liu2020lightning} and \cite{zhao2024soci}, TRUST employs the (2, 2)-threshold Paillier cryptosystem and integrates software and hardware to maximize the security offered by cryptographic primitives and improve the computational efficiency of TEE leveraging CPU resources. TRUST seeks to overcome the real-world challenges in privacy-preserving computation, such as the limited storage resources of TEE and overly strong security assumptions (e.g., non-collusion among multiple servers and the absolute security of TEE), which are misaligned with practical scenarios.

\section{Preliminarys}\label{Preliminarys}
In this section, we introduce the threshold Paillier cryptosystem and TEE, which serve as the foundation for the design of TRUST.

\subsection{FastPaiTD}
FastPaiTD proposed by the work \cite{zhao2024soci}, a $(2, 2)$-threshold Paillier cryptosystem, serves as the preliminarys of our proposed TRUST. FastPaiTD involves 5 probabilistic polynomial time algorithms listed as follows.

1) \kgen\ (key generation): \kgen\ takes a secure parameter $\kappa$ (e.g., $\kappa = 112$) as an input and outputs a public/private pair $(pk, sk)$ and two partially private keys $(sk_1, sk_2)$. Formally, $\kgen(1^\kappa) \rightarrow (pk, sk, sk_1, sk_2)$. Specifically, $pk = (N, h)$, $sk = \alpha$, and $(sk_1, sk_2)$ satisfy $sk_1 + sk_2 = 0 \mod 2\alpha$ and $sk_1 + sk_2 = 1 \mod N$. More details about the above parameters can refer to the work \cite{zhao2024soci}. Note that we set $sk_1$ as a $\sigma$-bit (e.g., $\sigma = 128$) random number and compute $sk_2 = 2\alpha \cdot (2\alpha)^{-1} \mod N - sk_1$. 

2) \enc\ (encryption): \enc\ takes as input a message $m \in \mathbb{Z}_N$ and the public key $pk$, and outputs a ciphertext $\lb m \rb$ as follows
\begin{equation}
    \label{eq:enc}
    \lb m \rb \leftarrow (1 + N)^m \cdot (h^r \md N) ^ N \md N^2,
\end{equation}
where $r$ is a $4\kappa$-bit random number. Formally, $\texttt{Enc}(pk, m) \rightarrow \lb m \rb$.

3) \dec\ (decryption): \dec\ takes as input $\lb m \rb$ adn the private key $sk$ and outputs a message $m$ as follows
\begin{equation}
    m \leftarrow (\frac{\lb m \rb^{2\alpha} \md N^2 - 1}{N} \md N ) \cdot (2\alpha)^{-1} \md N.
\end{equation}
Formally, $\texttt{Dec}(sk, \lb m \rb) \rightarrow m$.

4) \pdec (partial decryption): \pdec takes as input $\lb m \rb$ and a partially private key $sk_i$ ($i \in \{1, 2\}$) and outputs a ciphertext as follows
\begin{equation}
    [m]_i \leftarrow \lb m \rb^{sk_i} \mod N^2.
\end{equation}
Formally, $\mathsf{PDec}(sk_i, \lb m \rb) \rightarrow [m]_i$ for $i \in \{1, 2\}$.

5) \tdec (threshold decryption): \tdec takes as input $[m]_1$ and $[m]_2$ and outputs a message $m$ as follows
\begin{equation}
    m \leftarrow \frac{([m]_1 \cdot [m]_2 \md N^2) - 1}{N} \md N.
\end{equation}
Formally, $\mathsf{TDec}([m]_1, [m]_2) \rightarrow m$.

Note that FastPaiTD supports additive homomorphism and scalar-multiplication homomorphism. Specifically, 
\begin{itemize}
    \item Additive homomorphism: Given two ciphertexts $\lb m_1 \rb$ and $\lb m_2 \rb$, $\dec(sk, \lb m_1 + m_2 \mod N\rb = \dec(sk, \lb m_1 \rb \cdot \lb m_2 \rb \mod N^2)$ holds.
    \item Scalar-multiplication homomorphism: Given a ciphertext $\lb m \rb$ and a constant $c \in \mathbb{Z}_N$, $\dec(sk, \lb c \cdot m \mod N \rb \mod N^2) = \dec(sk, \lb m \rb^c \mod N^2)$ holds. Particularly, when $c = -1$, $\dec(sk, \lb m \rb^c \mod N^2) = -m$ holds.
\end{itemize}

\subsection{Trusted Execution Environment (TEE)}
TEE performs confidential computing by providing an isolated secure enclave within the processor, specifically designed to securely execute user-defined sensitive code in an untrusted host environment. More importantly, TEE provides robust protection for data and computation against any potentially malicious entities residing in the system (including the kernel, hypervisor, etc.) \cite{li2023survey}. Notably, TEE-based secure computation has attracted significant interest from researchers in both academia and industry over the past two decades (e.g., \cite{lie2000architectural}, \cite{champagne2010scalable}, \cite{maas2013phantom}, \cite{costan2016sanctum}), thanks to its several essential features.

\begin{itemize}
\item Attestation \cite{anati2013innovative}. 
At its core, attestation establishes trust between an enclave and another enclave or a remote machine through a secure communication channel. It includes two mechanisms: local attestation, enabling mutual verification between enclaves on the same platform, and remote attestation, which allows clients to verify the integrity and configuration of their program in a remote server.

\item Isolated Execution \cite{mckeen2013innovative}.
Applications leveraging TEE are divided into trusted and untrusted regions, with code and data in the trusted region executing within TEE, offering both confidentiality and integrity guarantees, while those in the untrusted region operate in REE and remain accessible to the operating system. User-defined programs are loaded into enclaves for isolated execution, where sensitive data is securely maintained in Enclave Page Caches (EPCs) within Processor Reserved Memory (PRM), which is pre-configured in Dynamic Random Access Memory (DRAM) during the bootstrapping phase. The code loaded inside an enclave can access both PRM and non-PRM memory, whereas external programs are restricted from accessing PRM. Programmer-defined entry points facilitate secure and adaptable communication between TEE and REE, enabling controlled interaction across these regions.

\item Sealing and Unsealing \cite{anati2013innovative}.
Sealing and unsealing is a mechanism that allows a TEE to securely store and retrieve a secret in non-volatile external storage. On one hand, sealing uses a CPU-derived key, enabling TEE to bind the secret to a specific enclave identity, then encrypt and store it externally. On the other hand, unsealing enables TEE to retrieve and decrypt the secret from storage.
\end{itemize}

While TEE safeguards processed data by encrypting both incoming and outgoing data flows, research has shown that it remains vulnerable to various side-channel attacks, including those targeting page tables \cite{xu2015controlled}, \cite{wang2017leaky}, DRAM \cite{chen2019sgxpectre}, \cite{lipp2020meltdown}, \cite{van2018foreshadow}, and cache \cite{liu2015last}, \cite{gotzfried2017cache}, \cite{moghimi2017cachezoom}.

\begin{figure}[t]
    \centering
    \includegraphics[width=0.98\linewidth]{ 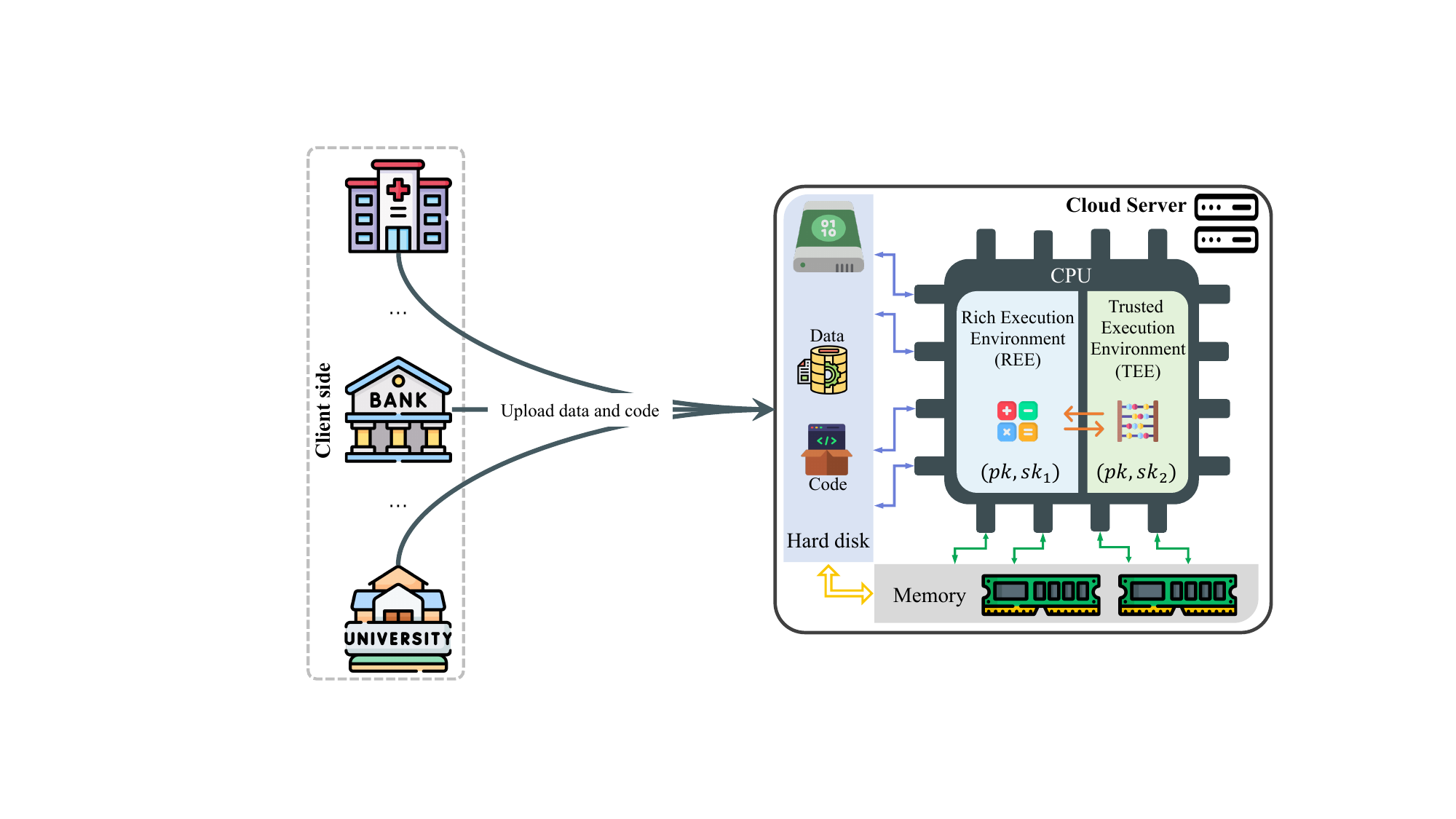}
    \caption{System model.}
    \label{fig:sysmod}
\end{figure}

\section{System Model and Threat Model}\label{System Model and Threat Model}
The system model formulates entities in a system and their abilities. As decipted in Fig. \ref{fig:sysmod}, our system model contains two entities, i.e., a client and a cloud server, where the cloud server is equipped a TEE chip. 

\begin{itemize}
    \item \textbf{Client.} The client is a user who has the requirement of SOC over ciphertexts. The user may be a hospital, a bank, or an university and so on. To this end, the user initializes the $(2, 2)$-threshold Paillier cryptosystem and generates the public key $pk$ and two partially private keys $sk_1$ and $sk_2$. Additionally, the user assigns $(pk, sk_1)$ to REE of the cloud server, and transmits $(pk, sk_2)$ to the TEE of the cloud server via a secure channel. The user stores ciphertexts and code on the hard disk of the cloud server.
    \item \textbf{Cloud server.} The cloud server is required to equip a TEE on its CPU, and provides SOC services for the user through its REE and TEE. Specifically, REE initiates a calculation and obtain the calculation result from TEE through cooperating with TEE.
\end{itemize}

The threat model assumes the attack capability of entities in the system model. In TRUST, we suppose the client is trusted. The REE and TEE of the cloud server are semi-honest, which means either REE or TEE follows protocols but tries to obtain the client's data and calculation results. Particularly, any adversary is allowed to corrupt REE or TEE, but cannot corrupt REE and TEE simultaneously.

\section{TRUST Design}\label{TRUST Design}
In this section, we firstly formalize the SOC of TRUST. After that, a suite of SOC protocols integrated into TRUST is introduced.

\subsection{Problem Formulation}
Suppose encrypted outsourced data $\lb m_1 \rb, \lb m_2 \rb, \cdots, \lb m_n \rb$ require to perform SOC denoted by $\mathcal{F}$, where $n$ is the number of operands and $n \geq 1$. Our proposed TRUST assumes $f_r$ and $f_t$ denote the computation over REE and TEE, respectively. Particularly, TRUST formalizes the SOC $\mathcal{F}(\lb m_1 \rb, \cdots, \lb m_n \rb)$ as follows.
\begin{definition}
Suppose a SOC be denoted by $\mathcal{F}(\lb m_1 \rb, \cdots, \lb m_n \rb)$, where $n$ is the size of outsourced data, TRUST formulates it as a two-party computation between REE and TEE denoted by
\begin{equation}
    \label{trust_soc}
    \mathcal{F}(\lb m_1 \rb, \cdots, \lb m_n \rb) \!\!\!\iff\!\!\! f_t(key_t, f_r(key_r, \lb m_1 \rb, \cdots, \lb m_n \rb)).
\end{equation}
$key_r$ and $key_t$ indicate the key of REE and TEE, respectively. Note that $f_t$ can be omitted when $f_r$ enables computations required by $\mathcal{F}$.
\end{definition}
From Eq. (\ref{trust_soc}), we see that only REE accesses outsourced data $\lb m_1 \rb, \cdots, \lb m_n \rb$ directly, while TEE takes the output of REE as an input. According to the system model of TRUST, $key_r$ and $key_t$ denote $(pk, sk_1)$ and $(pk, sk_2)$, respectively, thus, the output of $f_r$ is encrypted data. Note that when $n = 1$, $\mathcal{F}(\lb m_1 \rb)$ means a unary operation (e.g., absolute value, factorial). $\mathcal{F}(\lb m_1 \rb, \lb m_2 \rb)$ ($n = 2$) indicates a binary operation (e.g., addition, multiplication, comparison), while $\mathcal{F}(\lb m_1 \rb, \lb m_2 \rb, \lb m_3 \rb)$ ($n = 3$) denotes a ternary operation (e.g., ternary conditional operator), and so on.

Fig. \ref{fig:workflow} shows a typical workflow of TRUST. As depicted in Fig. \ref{fig:workflow}, TRUST reconstructs the computation $\mathcal{F}(\lb m_1 \rb, \cdots, \lb m_n \rb)$ into $f_r(pk, sk_1, \lb m_1 \rb, \cdots, \lb m_n \rb) \rightarrow \lb R_0 \rb$ and $f_t(pk, sk_2, \lb R_0 \rb) \rightarrow \mathcal{F}(\lb m_1 \rb, \cdots, \lb m_n \rb)$ (Case 1), or $f_r(pk, \lb m_1 \rb, \cdots, \lb m_n \rb)$ (Case 2), where $\lb R_0 \rb$ is an intermediate result, and $n$ means the amount of operands.
\begin{figure}[ht]
    \centering
    \includegraphics[width=0.98\linewidth]{ 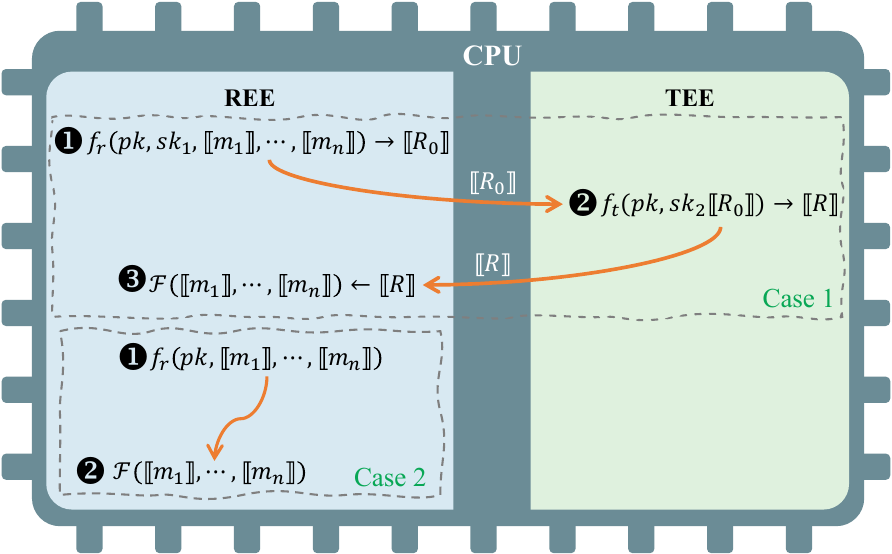}
    \caption{Typical workflow of TRUST.}
    \label{fig:workflow}
\end{figure}

\subsection{Protocols Design}
In this section, we detail operations including multiplication, comparison, equality, absolute value, ternary conditional operator, etc. supported by TRUST. Obviously, TRUST is easy to achieve addition and scalar-multiplication as they are homomorphic features of FastPaiTD. Take addition as an example, $\mathcal{F}(\lb m_1 \rb, \lb m_2 \rb)$ is formulated as $\mathcal{F}(\lb m_1 \rb, \lb m_2 \rb) \rightarrow \lb m_1 + m_2 \rb$. In this case, $f_r$ is formulated as $f_r(\lb m_1 \rb, \lb m_2 \rb) = \lb m_1 \rb \cdot \lb m_2 \rb \mod N^2$, while $f_t$ can be omitted. Following the prior work \cite{zhao2023soci}, we assume $m_1, m_2, \cdots, m_n \in (-2^\ell, 2^\ell)$, where $\ell$ is a constant (e.g., $\ell = 32$).

Algorithm \ref{alg:mul} shows a secure multiplication algorithm following Definition 1. Formally, $\mathcal{F}(\lb m_1 \rb, \cdots, \lb m_n \rb)$ is denoted by $\mathcal{F}_{\mathsf{mul}}(\lb m_1 \rb, \lb m_2 \rb) \to \lb m_1 \cdot m_2 \rb$, where $n = 2$. At Step 1, REE takes as input keys ($pk$ and $sk_1$), $\lb m_1 \rb$, and $\lb m_2 \rb$, and outputs $\lb m_1 + r \rb$, $[m_1 + r]_1$, $\lb m_2 \rb$, and $\lb -r \cdot m_2\rb$, where $r$ is a $\sigma$-bit (e.g., $\sigma = 128$) random number, and $\sample$ means random sampling in this work.

After receiving the output of REE, TEE firstly obtains $m_1 + r$ by calling \pdec and \tdec sequentially and then outputs the final result $\lb R \rb$. According to the homomorphic properties of FastPaiTD, we have
\begin{equation}
    \lb m_2 \rb^{m_1 + r} \cdot \lb -r \cdot m_2 \rb \cdot \enc(pk, 0) \to \lb m_1 \cdot m_2 + r \cdot m_2 - r \cdot m_2 +0 \rb.
\end{equation}
Thus, it is easy to verify $R = m_1 \cdot m_2$. In other words, $\mathcal{F}_{\mathsf{mul}}$ outputs $\lb m_1 \cdot m_2 \rb$ correctly.
\begin{algorithm}[htbp]
	\caption{
        $\mathcal{F}_{\mathsf{mul}}(\lb m_1 \rb, \lb m_2 \rb) \to \lb m_1 \cdot m_2 \rb$}
        \vspace{2pt}
        \label{alg:mul}
	\KwIn{$\lb m_1 \rb$ and $\lb m_2 \rb$.}
        \KwOut{$\lb m_1 \cdot m_2 \rb$.}

        \vspace{3pt}

        \begin{mdframed}[backgroundcolor=reecolor,innerleftmargin=2pt,innerrightmargin=2pt,innerbottommargin=2pt,leftmargin=-8pt,rightmargin=12pt]
            $\triangleright$ Step 1: Computations in REE 
        \begin{mdframed}[backgroundcolor=whitecolor,innertopmargin=4pt,innerbottommargin=4pt,leftmargin=0pt,rightmargin=1pt,innerleftmargin=-4pt]
            \quad $\bullet$ $\lb m_1 + r \rb \gets \lb m_1 \rb \cdot \enc(pk, r)$, where $r \sample \bin^{\sigma}$;
            
            \quad $\bullet$ $\lb -r \cdot m_2 \rb \gets \lb m_2 \rb^{-r}$;
            
            \quad $\bullet$ $\mathsf{PDec}(sk_1, \lb m_1 + r \rb) \to [m_1 + r]_1$;
        \end{mdframed}
        \end{mdframed}
        
        \vspace{-8pt}
        
        \begin{mdframed}[backgroundcolor=whitecolor,innerleftmargin=4pt,innertopmargin=0pt,innerrightmargin=2pt,innerbottommargin=-2pt,leftmargin=-8pt,rightmargin=12pt]
            REE \quad {$\sendmessageright*{\langle \lb m_1 + r \rb, [m_1 + r]_1, \lb m_2 \rb, \lb -r \cdot m_2\rb \rangle}$} TEE
        \end{mdframed}
        
        \vspace{-10pt}
        
        \begin{mdframed}[backgroundcolor=teecolor,innerleftmargin=2pt,innerrightmargin=2pt,innerbottommargin=2pt,leftmargin=-8pt,rightmargin=12pt]
             $\triangleright$ Step 2: Computations in TEE
        \begin{mdframed}[backgroundcolor=whitecolor,innertopmargin=4pt,innerbottommargin=4pt,leftmargin=0pt,rightmargin=1pt,innerleftmargin=-4pt]
            \quad $\bullet$ $\mathsf{PDec}(sk_2, \lb m_1 + r \rb) \to [m_1 + r]_2$;

            \quad $\bullet$ $\mathsf{TDec}([m_1 + r]_1, [m_1 + r]_2) \to m_1 + r$;

            \quad $\bullet$ $\lb R \rb \gets \lb m_2 \rb^{m_1 + r} \cdot \lb -r \cdot m_2 \rb \cdot \enc(pk, 0)$;
        \end{mdframed}
        \end{mdframed}

        \vspace{-8pt}

        \begin{mdframed}[backgroundcolor=whitecolor,innerleftmargin=4pt,innertopmargin=0pt,innerrightmargin=2pt,innerbottommargin=-2pt,leftmargin=-8pt,rightmargin=12pt]
            REE \qquad $\sendmessageleft*{\lb R \rb}$ TEE
        \end{mdframed}

        \vspace{-10pt}
        
        \begin{mdframed}[backgroundcolor=reecolor,innerleftmargin=2pt,innerrightmargin=2pt,innerbottommargin=2pt,leftmargin=-8pt,rightmargin=12pt]
            $\triangleright$ Step 3: Computations in REE 
        \begin{mdframed}[backgroundcolor=whitecolor,innertopmargin=4pt,innerbottommargin=4pt,leftmargin=0pt,rightmargin=1pt,innerleftmargin=-4pt]
            \quad $\bullet$ $\mathcal{F}_{\mathsf{mul}} \gets \lb R \rb$;
        \end{mdframed}
        \end{mdframed}
        \vspace{-5pt}
\end{algorithm}

Algorithm \ref{alg:cmp} depicts a secure comparison algorithm following Definition 1. Formally, $\mathcal{F}(\lb m_1 \rb, \cdots, \lb m_n \rb)$ is denoted by $\mathcal{F}_{\mathsf{cmp}}(\lb m_1 \rb, \lb m_2 \rb) \to \lb \mu \rb$, where $\mu = 0$ when $m_1 \geq m_2$, otherwise ($m_1 < m_2$), $\mu = 1$. At Step 1, REE takes as input keys ($pk$ and $sk_1$), $\lb m_1 \rb$, and $\lb m_2 \rb$, and outputs $\lb d \rb$, $[d]_1$, and $\lb \pi \rb$. Specifically, REE firstly generates $r_1$ and $r_2$ ($r_1 + r_2 > \frac{N}{2}$) by random sampling. After that, REE randomly samples $\pi$ from $\bin$ and computes $\lb d \rb$. According to the homomorphic properties of FastPaiTD, we have 
\begin{equation}
    \label{eq:dv}
    d = \begin{cases}
	     r_1 \cdot (m_1 - m_2 + 1) + r_2, &\pi = 0; \\
	     r_1 \cdot (m_2 - m_1) + r_2. &\pi = 1;	\end{cases}
\end{equation}

After receiving the output of REE, TEE firstly obtains $d$ by calling \pdec and \tdec sequentially. Then, TEE computes and returns $\lb R \rb$ to REE.
It is easy to verify
\begin{equation}
    \label{eq:cmpr}
    R = \mu_0 + (1 - 2 \mu_0) \cdot \pi.
\end{equation}
If $\mu_0 = 0$ (i.e., $d > \frac{N}{2}$), we have $R = \pi$, whereas $\mu_0 = 1$ (i.e., $d < \frac{N}{2}$), we have $R = 1 - \pi$. As shown in Eq. (\ref{eq:cmpr}), the value of $R$ depends on the ones $\mu_0$ and $\pi$, while the one $\mu_0$ depends on that $d$.

\textit{Case 1} ($\pi = 0$): If $r_1 \cdot (m_1 - m_2 + 1) + r_2 > \frac{N}{2}$ (i.e., $d > \frac{N}{2}$), we have $(m_1 - m_2 + 1) \geq 1$ because $r_1 + r_2 > \frac{N}{2}$ and $r_2 < \frac{N}{2}$. In this case, $m_1 \geq m_2$, while $R = \mu_0 = 0$. 
If $r_1 \cdot (m_1 - m_2 + 1) + r_2 < \frac{N}{2}$ (i.e., $d < \frac{N}{2}$), we have $(m_1 - m_2 + 1) \leq 0$ because $r_1 + r_2 > \frac{N}{2}$ and $r_2 < \frac{N}{2}$. In this case, $m_1 < m_2$, while $R = \mu_0 = 1$.
In other words, when $\pi = 0$, if $m_1 \geq m_2$, $\mathcal{F}_{\mathsf{cmp}}$ outputs $\mu = 0$, otherwise ($m_1 < m_2$), outputs $\mu = 1$.

\textit{Case 2} ($\pi = 1$): If $r_1 \cdot (m_2 - m_1) + r_2 > \frac{N}{2}$ (i.e., $d > \frac{N}{2}$), we have $m_2 - m_1 \geq 1$ because $r_1 + r_2 > \frac{N}{2}$ and $r_2 < \frac{N}{2}$. In this case, $m_1 < m_2$, while $R = 1 - \mu_0 = 1$. 
If $r_1 \cdot (m_2 - m_2) + r_2 < \frac{N}{2}$ (i.e., $d < \frac{N}{2}$), we have $m_2 - m_1 \leq 0$ because $r_1 + r_2 > \frac{N}{2}$ and $r_2 < \frac{N}{2}$. In this case, $m_1 \geq m_2$, while $R = 1 - \mu_0 = 0$.
In other words, when $\pi = 1$, if $m_1 \geq m_2$, $\mathcal{F}_{\mathsf{cmp}}$ outputs $\mu = 0$, otherwise ($m_1 < m_2$), outputs $\mu = 1$.

Taken together, either $\pi = 0$ or $\pi = 1$, $\mathcal{F}_{\mathsf{cmp}}$ always computes $\mu$ correctly.
\begin{algorithm}[htbp]
	\caption{
        $\mathcal{F}_{\mathsf{cmp}}(\lb m_1 \rb, \lb m_2 \rb) \to \lb \mu \rb$}
        \vspace{2pt}
        \label{alg:cmp}
	\KwIn{$\lb m_1 \rb$ and $\lb m_2 \rb$.}
        \KwOut{$\lb \mu \rb$, if $m_1 \geq m_2$, $\mu = 0$, otherwise, $\mu = 1$.}

        \vspace{3pt}

        \begin{mdframed}[backgroundcolor=reecolor,innerleftmargin=2pt,innerrightmargin=2pt,innerbottommargin=2pt,leftmargin=-8pt,rightmargin=12pt]
            $\triangleright$ Step 1: Computations in REE 
        \begin{mdframed}[backgroundcolor=whitecolor,innertopmargin=4pt,innerbottommargin=4pt,leftmargin=0pt,rightmargin=1pt,innerleftmargin=-4pt]
            \quad $\bullet$ $r_1 \sample \bin^{\sigma}$ and $r_2 \sample (\frac{N}{2} - r_1, \frac{N}{2})$;
            
            \quad $\bullet$ $\pi \sample \bin$, \hangindent = 18pt 
            \begin{equation}
                \nonumber
	        \lb d \rb \gets \begin{cases}
	            \lb m_1 \rb^{r_1} \cdot \lb m_2 \rb^{-r_1} \cdot \enc(pk, r_1 + r_2), &\pi = 0; \\
	            \lb m_1 \rb^{-r_1} \cdot \lb m_2 \rb^{r_1} \cdot \enc(pk, r_2), &\pi = 1;	
	        \end{cases}
            \end{equation}
            
            \quad $\bullet$ $\mathsf{PDec}(sk_1, \lb d \rb) \to [d]_1$ and $\enc(pk, \pi) \to \lb \pi \rb$;
        \end{mdframed}
        \end{mdframed}
        
        \vspace{-8pt}
        
        \begin{mdframed}[backgroundcolor=whitecolor,innerleftmargin=4pt,innertopmargin=0pt,innerrightmargin=2pt,innerbottommargin=-2pt,leftmargin=-8pt,rightmargin=12pt]
            REE \quad {$\sendmessageright*{\langle \lb d \rb, [d]_1, \lb \pi \rb \rangle}$} TEE
        \end{mdframed}
        
        \vspace{-10pt}
        
        \begin{mdframed}[backgroundcolor=teecolor,innerleftmargin=2pt,innerrightmargin=2pt,innerbottommargin=2pt,leftmargin=-8pt,rightmargin=12pt]
             $\triangleright$ Step 2: Computations in TEE
        \begin{mdframed}[backgroundcolor=whitecolor,innertopmargin=4pt,innerbottommargin=4pt,leftmargin=0pt,rightmargin=1pt,innerleftmargin=-4pt]
            \quad $\bullet$ $\mathsf{PDec}(sk_2, \lb d \rb) \to [d]_2$ and $\mathsf{TDec}([d]_1, [d]_2) \to d$;

            \hangindent = -36pt
            \vspace{-12pt}
            \begin{flalign*}
	        \bullet \begin{cases}
	            \mu_0 = 0 \mathrm{\ and\ } \lb \mu_0 \rb \gets \enc(pk, 0), &d > \frac{N}{2}; \\
	            \mu_0 = 1 \mathrm{\ and\ } \lb \mu_0 \rb \gets \enc(pk, 1), &d < \frac{N}{2};	
	        \end{cases}
	        \end{flalign*}

            \quad $\bullet$ $\lb R \rb \gets \lb \mu_0 \rb \cdot \lb \pi \rb^{1-2\mu_0}$;
        \end{mdframed}
        \end{mdframed}

        \vspace{-8pt}

        \begin{mdframed}[backgroundcolor=whitecolor,innerleftmargin=4pt,innertopmargin=0pt,innerrightmargin=2pt,innerbottommargin=-2pt,leftmargin=-8pt,rightmargin=12pt]
            REE \qquad $\sendmessageleft*{\lb R \rb}$ TEE
        \end{mdframed}

        \vspace{-10pt}
        
        \begin{mdframed}[backgroundcolor=reecolor,innerleftmargin=2pt,innerrightmargin=2pt,innerbottommargin=2pt,leftmargin=-8pt,rightmargin=12pt]
            $\triangleright$ Step 3: Computations in REE 
        \begin{mdframed}[backgroundcolor=whitecolor,innertopmargin=4pt,innerbottommargin=4pt,leftmargin=0pt,rightmargin=1pt,innerleftmargin=-4pt]
            \quad $\bullet$ $\mathcal{F}_{\mathsf{cmp}} \gets \lb R \rb$;
        \end{mdframed}
        \end{mdframed}
        \vspace{-5pt}
\end{algorithm}

Algorithms \ref{alg:mul} and \ref{alg:cmp} denote a linear secure operation and a non-linear secure operation, respectively. In addition to the non-linear comparison operation, a equality operation is also common. Algorithm \ref{alg:eql} lists the detail of a secure equality algorithm. Formally, $\mathcal{F}(\lb m_1 \rb, \cdots, \lb m_n \rb)$ is denoted by $\mathcal{F}_{\mathsf{eql}}(\lb m_1 \rb, \lb m_2 \rb) \to \lb \mu \rb$, where $\mu = 0$ if $m_1 = m_2$ holds, otherwise ($m_1 \neq m_2$), $\mu = 1$. $\mathcal{F}_{\mathsf{eql}}$ is concreted to two critical steps. At Step 1, REE computes $\lb d_1 \rb$ and $\lb d_2 \rb$ 
in the same way as $\lb d \rb$ in Algorithm \ref{alg:cmp}, and outputs $\lb d_1 \rb$, $[d_1]_1$, $\lb \pi_1 \rb$, $\lb d_2 \rb$, $[d_2]_1$, and $\lb \pi_2 \rb$. According to the homomorphic properties of FastPaiTD, we have 
\begin{equation}
    \label{eq:d1v}
    d_1 = \begin{cases}
	     r_1 \cdot (m_1 - m_2 + 1) + r_2, &\pi = 0; \\
	     r_1 \cdot (m_2 - m_1) + r_2, &\pi = 1;	\end{cases}
\end{equation}
\begin{equation}
    \label{eq:d2v}
    d_2 = \begin{cases}
	     r_1' \cdot (m_2 - m_1 + 1) + r_2', &\pi = 0; \\
	     r_1' \cdot (m_1 - m_2) + r_2'. &\pi = 1;	\end{cases}
\end{equation}

At Step 2, TEE obtains $d_1$ and $d_2$ by calling \pdec and $\mathsf{TDec}$. Essentially, $d_1$ and $d_2$ implies the magnitude relationship between $m_1$ and $m_2$, and the one between $m_2$ and $m_1$, respectively. Mathematically, $m_1 = m_2$ is equivalent to $m_1 \geq m_2 \And m_2 \geq m_1$. Inspired by this idea, at Step 2, TEE computes $\lb R \rb$, where 
\begin{equation}
    \label{eq:eqlr}
    R = \mu_0 + (1 - 2\mu_0) \cdot \pi_1 + \mu_0' + (1 - 2\mu_0') \cdot \pi_2.
\end{equation}

According to Algorithm \ref{alg:cmp}, if $m_1 \geq m_2$, $\mu_0 + (1 - 2\mu_0) \cdot \pi_1 = 0$. Additionally, if $m_2 \geq m_1$, $\mu_0' + (1 - 2\mu_0') \cdot \pi_2 = 0$. And since whether it is $\mu_0 + (1 - 2\mu_0) \cdot \pi_1$ or $\mu_0' + (1 - 2\mu_0') \cdot \pi_2$, it can only result in $0$ and $1$. Furthermore, $\mu_0 + (1 - 2\mu_0) \cdot \pi_1$ and $\mu_0' + (1 - 2\mu_0') \cdot \pi_2$ cannot both are equal to $1$ as $m_1 < m_2 \And m_2 < m_1$ cannot be true at the same time. Therefore, $R = 0$ if and only if $m_1 \geq m_2 \And m_2 \geq m_1$ (i.e., $m_1 = m_2$). Also, $R = 1$ if and only if $m_1 > m_2 \And m_2 < m_1$ or $m_1 < m_2 \And m_2 > m_1$. In other words, if $m_1 = m_2$, $\mathcal{F}_{\mathsf{eql}}$ outputs $\lb 0 \rb$, otherwise $m_1 \neq m_2$ (i.e., $m_1 > m_2$ or $m_2 > m_1$), it outputs $\lb 1 \rb$. Thus, we say that $\mathcal{F}_{\mathsf{eql}}$ outputs $\lb \mu \rb$ correctly.
\begin{algorithm}[htbp]
	\caption{
        $\mathcal{F}_{\mathsf{eql}}(\lb m_1 \rb, \lb m_2 \rb) \to \lb \mu \rb$}
        \vspace{2pt}
        \label{alg:eql}
	\KwIn{$\lb m_1 \rb$ and $\lb m_2 \rb$.}
        \KwOut{$\lb \mu \rb$, if $m_1 = m_2$, $\mu = 0$, otherwise, $\mu = 1$.}

        \vspace{3pt}

        \begin{mdframed}[backgroundcolor=reecolor,innerleftmargin=2pt,innerrightmargin=2pt,innerbottommargin=2pt,leftmargin=-8pt,rightmargin=12pt]
            $\triangleright$ Step 1: Computations in REE 
        \begin{mdframed}[backgroundcolor=whitecolor,innertopmargin=4pt,innerbottommargin=4pt,leftmargin=0pt,rightmargin=1pt,innerleftmargin=-4pt]
            \quad $\bullet$ $r_1 \sample \bin^{\sigma}, r_2 \sample (\frac{N}{2} - r_1, \frac{N}{2})$ and\\\hangindent = 18pt
            $r_1' \sample \bin^{\sigma}, r_2' \sample (\frac{N}{2} - r_1', \frac{N}{2})$;
            
            \quad $\bullet$ $\pi_1 \sample \bin$ and $\pi_2 \sample \bin$, \hangindent = 18pt 
            \begin{equation}
                \nonumber
	        \lb d_1 \rb \!\gets\! \begin{cases}
	            \lb m_1 \rb^{r_1} \cdot \lb m_2 \rb^{-r_1} \!\cdot\! \enc(pk, r_1 + r_2), \!\!\!&\pi_1 = 0; \\
	            \lb m_1 \rb^{-r_1} \cdot \lb m_2 \rb^{r_1} \!\cdot\! \enc(pk, r_2), \!\!\!&\pi_1 = 1;	
	        \end{cases}
            \end{equation}
            \hangindent = 18pt 
            \begin{equation}
                \nonumber
	        \lb d_2 \rb \!\gets\! \begin{cases}
	            \lb m_2 \rb^{r_1'} \cdot \lb m_1 \rb^{-r_1'} \!\cdot\! \enc(pk, r_1' + r_2'), \!\!\!&\pi_2 = 0; \\
	            \lb m_2 \rb^{-r_1'} \cdot \lb m_1 \rb^{r_1'} \!\cdot\! \enc(pk, r_2'), \!\!\!&\pi_2 = 1;	
	        \end{cases}
            \end{equation}
            
            \quad $\bullet$ $\mathsf{PDec}(sk_1, \lb d_1 \rb) \to [d_1]_1, \enc(pk, \pi_1) \to \lb \pi_1 \rb$ and \\ \hangindent = 18pt
            $\mathsf{PDec}(sk_1, \lb d_2 \rb) \to [d_2]_1, \enc(pk, \pi_2) \to \lb \pi_2 \rb$;
        \end{mdframed}
        \end{mdframed}
        
        \vspace{-8pt}
        
        \begin{mdframed}[backgroundcolor=whitecolor,innerleftmargin=4pt,innertopmargin=0pt,innerrightmargin=2pt,innerbottommargin=-2pt,leftmargin=-8pt,rightmargin=12pt]
            REE \quad {$\sendmessageright*{\langle \lb d_1 \rb, [d_1]_1, \lb \pi_1 \rb, \lb d_2 \rb, [d_2]_1, \lb \pi_2 \rb \rangle}$} TEE
        \end{mdframed}
        
        \vspace{-10pt}
        
        \begin{mdframed}[backgroundcolor=teecolor,innerleftmargin=2pt,innerrightmargin=2pt,innerbottommargin=2pt,leftmargin=-8pt,rightmargin=12pt]
             $\triangleright$ Step 2: Computations in TEE
        \begin{mdframed}[backgroundcolor=whitecolor,innertopmargin=4pt,innerbottommargin=4pt,leftmargin=0pt,rightmargin=1pt,innerleftmargin=-4pt]
            \quad $\bullet$ $\mathsf{PDec}(sk_2, \lb d_1 \rb) \to [d_1]_2, \mathsf{TDec}([d_1]_1, [d_1]_2) \to d_1$ and \\\hangindent=18pt
            $\mathsf{PDec}(sk_2, \lb d_2 \rb) \to [d_2]_2, \mathsf{TDec}([d_2]_1, [d_2]_2) \to d_2$;

            \hangindent = -36pt
            \vspace{-12pt}
            \begin{flalign*}
	        \bullet \begin{cases}
	            \mu_0 = 0 \mathrm{\ and\ } \lb \mu_0 \rb \gets \enc(pk, 0), &d_1 > \frac{N}{2}; \\
	            \mu_0 = 1 \mathrm{\ and\ } \lb \mu_0 \rb \gets \enc(pk, 1), &d_1 < \frac{N}{2};\\
	            \mu_0' = 0 \mathrm{\ and\ } \lb \mu_0' \rb \gets \enc(pk, 0), &d_2 > \frac{N}{2}; \\
	            \mu_0' = 1 \mathrm{\ and\ } \lb \mu_0' \rb \gets \enc(pk, 1), &d_2 < \frac{N}{2};	
	        \end{cases}
	        \end{flalign*}

            \quad $\bullet$ $\lb R \rb \gets \lb \mu_0 \rb \cdot \lb \pi_1 \rb^{1-2\mu_0} \cdot \lb \mu_0' \rb \cdot \lb \pi_2 \rb^{1-2\mu_0'}$;
        \end{mdframed}
        \end{mdframed}

        \vspace{-8pt}

        \begin{mdframed}[backgroundcolor=whitecolor,innerleftmargin=4pt,innertopmargin=0pt,innerrightmargin=2pt,innerbottommargin=-2pt,leftmargin=-8pt,rightmargin=12pt]
            REE \qquad $\sendmessageleft*{\lb R \rb}$ TEE
        \end{mdframed}

        \vspace{-10pt}
        
        \begin{mdframed}[backgroundcolor=reecolor,innerleftmargin=2pt,innerrightmargin=2pt,innerbottommargin=2pt,leftmargin=-8pt,rightmargin=12pt]
            $\triangleright$ Step 3: Computations in REE 
        \begin{mdframed}[backgroundcolor=whitecolor,innertopmargin=4pt,innerbottommargin=4pt,leftmargin=0pt,rightmargin=1pt,innerleftmargin=-4pt]
            \quad $\bullet$ $\mathcal{F}_{\mathsf{eql}} \gets \lb R \rb$;
        \end{mdframed}
        \end{mdframed}
        \vspace{-5pt}
\end{algorithm}

Now, we discuss one unary operation (i.e., absolute value). If $m \geq 0$, $|m| = m$, otherwise ($m < 0$), $|m| = -m$. Thus, $|m|$ can be denoted by $|m| = (1 - 2\mu) \cdot m$, where $\mu = 0$ when $m \geq 0$, otherwise ($m < 0$), $\mu = 1$.

Inspired by the above observation, Algorithm \ref{alg:abs} depicts a secure absolute value algorithm. Formally, $\mathcal{F}(\lb m_1 \rb)$ is denoted by $\mathcal{F}_{\mathsf{abs}}(\lb m_1 \rb) \to \lb |m| \rb$, where $n$ is set as $n = 1$. Specifically, at Step 1, REE takes as input keys and $\lb m_1 \rb$, and outputs four ciphertexts including $\lb d \rb$, $[d]_1$, $\lb m_1 \rb$, and $\lb \pi \cdot m_1 \rb$, where $\lb \pi \cdot m_1 \rb \gets \lb m_1 \rb^\pi \cdot \enc(pk, 0)$. According to the homomorphic properties of FastPaiTD, we see
\begin{equation}
    \label{eq:absdv}
    d = \begin{cases}
	     r_1 \cdot (m_1 + 1) + r_2, &\pi = 0; \\
	     r_1 \cdot (- m_1) + r_2. &\pi = 1;	\end{cases}
\end{equation}

At Step 2, REE obtains $d$ by calling \pdec and $\mathsf{TDec}$. From Eq.(\ref{eq:absdv}), we see that $d$ essentially implies the magnitude relationship between $m_1$ and $0$. After that, TEE computes and returns $\lb R \rb$ to REE. It is easy to verify
\begin{equation}
    \label{eq:eqlr}
    R = (1 - 2\mu_0) \cdot m_1 + (4\mu_0 - 2) \cdot \pi \cdot m_1,
\end{equation}
where $\mu_0 = 0$ when $d > \frac{N}{2}$, otherwise, $\mu_0 = 1$.

As $|m_1| = (1 - 2\mu) \cdot m_1$, where $\mu = 0$ when $m_1 \geq 0$, otherwise, $\mu = 1$, according to Algorithm \ref{alg:cmp}, we see $\mu = \mu_0 + (1 - 2\mu_0) \cdot \pi$. Then, we have
\begin{align}
    |m_1| &= [1 - 2 \cdot (\mu_0 + (1 - 2\mu_0) \cdot \pi)] \cdot m_1 \nonumber\\
    &= (1 - 2\mu_0) \cdot m_1 + (4\mu_0 - 2) \cdot \pi \cdot m_1 \nonumber\\
    &= R.
\end{align}
Thus, we can say that $\mathcal{F}_{\mathsf{abs}}$ oputputs $\lb |m_1| \rb$ correctly.
\begin{algorithm}[htbp]
	\caption{
        $\mathcal{F}_{\mathsf{abs}}(\lb m_1 \rb) \to \lb |m_1| \rb$}
        \vspace{2pt}
        \label{alg:abs}
	\KwIn{$\lb m_1 \rb$.}
        \KwOut{$\lb |m_1| \rb$.}

        \vspace{3pt}

        \begin{mdframed}[backgroundcolor=reecolor,innerleftmargin=2pt,innerrightmargin=2pt,innerbottommargin=2pt,leftmargin=-8pt,rightmargin=12pt]
            $\triangleright$ Step 1: Computations in REE 
        \begin{mdframed}[backgroundcolor=whitecolor,innertopmargin=4pt,innerbottommargin=4pt,leftmargin=0pt,rightmargin=1pt,innerleftmargin=-4pt]
            \quad $\bullet$ $r_1 \sample \bin^{\sigma}$ and $r_2 \sample (\frac{N}{2} - r_1, \frac{N}{2})$;
            
            \quad $\bullet$ $\pi \sample \bin$, \hangindent = -25pt 
            \begin{equation}
                \nonumber
	        \lb d \rb \gets \begin{cases}
	            \lb m_1 \rb^{r_1} \cdot \enc(pk, r_1 + r_2), &\pi = 0; \\
	            \lb m_1 \rb^{-r_1} \cdot \enc(pk, r_2), &\pi = 1;	
	        \end{cases}
            \end{equation}
            
            \quad $\bullet$ $\mathsf{PDec}(sk_1, \lb d \rb) \to [d]_1$, $\lb \pi \cdot m_1 \rb \gets \lb m_1 \rb^\pi \cdot \enc(pk, 0)$;
        \end{mdframed}
        \end{mdframed}
        
        \vspace{-8pt}
        
        \begin{mdframed}[backgroundcolor=whitecolor,innerleftmargin=4pt,innertopmargin=0pt,innerrightmargin=2pt,innerbottommargin=-2pt,leftmargin=-8pt,rightmargin=12pt]
            REE \quad {$\sendmessageright*{\langle \lb d \rb, [d]_1, \lb m_1 \rb, \lb \pi \cdot m_1 \rb \rangle}$} TEE
        \end{mdframed}
        
        \vspace{-10pt}
        
        \begin{mdframed}[backgroundcolor=teecolor,innerleftmargin=2pt,innerrightmargin=2pt,innerbottommargin=2pt,leftmargin=-8pt,rightmargin=12pt]
             $\triangleright$ Step 2: Computations in TEE
        \begin{mdframed}[backgroundcolor=whitecolor,innertopmargin=4pt,innerbottommargin=4pt,leftmargin=0pt,rightmargin=1pt,innerleftmargin=-4pt]
            \quad $\bullet$ $\mathsf{PDec}(sk_2, \lb d \rb) \to [d]_2$ and $\mathsf{TDec}([d]_1, [d]_2) \to d$;

            \hangindent = -132pt
            \vspace{-12pt}
            \begin{flalign*}
	        \bullet\ \mu_0 = \begin{cases}
	            0, &d > \frac{N}{2}; \\
	            1, &d < \frac{N}{2};	
	        \end{cases}
	        \end{flalign*}

            \quad $\bullet$ $\lb R \rb \gets \lb m_1 \rb^{1 - 2\mu_0} \cdot \lb \pi \cdot m_1 \rb^{4\mu_0 - 2} \cdot \enc(pk, 0)$;
        \end{mdframed}
        \end{mdframed}

        \vspace{-8pt}

        \begin{mdframed}[backgroundcolor=whitecolor,innerleftmargin=4pt,innertopmargin=0pt,innerrightmargin=2pt,innerbottommargin=-2pt,leftmargin=-8pt,rightmargin=12pt]
            REE \qquad $\sendmessageleft*{\lb R \rb}$ TEE
        \end{mdframed}

        \vspace{-10pt}
        
        \begin{mdframed}[backgroundcolor=reecolor,innerleftmargin=2pt,innerrightmargin=2pt,innerbottommargin=2pt,leftmargin=-8pt,rightmargin=12pt]
            $\triangleright$ Step 3: Computations in REE 
        \begin{mdframed}[backgroundcolor=whitecolor,innertopmargin=4pt,innerbottommargin=4pt,leftmargin=0pt,rightmargin=1pt,innerleftmargin=-4pt]
            \quad $\bullet$ $\mathcal{F}_{\mathsf{abs}} \gets \lb R \rb$;
        \end{mdframed}
        \end{mdframed}
        \vspace{-5pt}
\end{algorithm}

Finally, we describe a secure ternary operation (i.e., ternary conditional operator). Formally, the ternary conditional operation can be formulated as ``\textit{if a then b else c}". Without loss of generality, $\mathcal{F}(\lb m_1 \rb, \lb m_2 \rb, \lb m_3 \rb)$ is denoted by $\mathcal{F}_{\mathsf{trn}}(\lb m_1 \rb, \lb m_2 \rb, \lb m_3 \rb) \to \lb m \rb$, where if $m_1 = 1$, then $m = m_2$, else $m = m_3$. Note that $m = m_2 + \mu \cdot (m_3 - m_2)$ is always true if $\mu = 0$ when $m_1 = 1$, otherwise ($m_1 \neq 1$), $\mu = 1$.

Algorithm \ref{alg:trn} shows the detail of the secure ternary operation. Specifically, at Step 1, REE computes $\lb d_1 \rb$ and $\lb d_2 \rb$ in the same way as those in Algorithm \ref{alg:eql}, and outputs eight ciphertexts, where $\lb \pi_1 \cdot (m_3 - m_2) \rb \gets \lb m_3 - m_2 \rb^{\pi_1} \cdot \enc(pk, 0)$, and $\lb \pi_2 \cdot (m_3 - m_2) \rb \gets \lb m_3 - m_2 \rb^{\pi_2} \cdot \enc(pk, 0)$.

At Step 2, TEE firstly obtains $d_1$ and $d_2$ by calling \pdec and $\mathsf{TDec}$. $d_1$ and $d_2$ essentially implies the magnitude relationship between $m_1$ and $1$ the one between $1$ and $m_1$, respectively. After that, TEE computes $\lb R \rb$. According to Eq. (\ref{eq:eqlr}), it is easy to verify that
\begin{align}
    R = & m_2 \!+\! (\mu_0 + \mu_0') \cdot (m_3 - m_2) \!+\! (1 \!-\! 2\mu_0) \cdot \pi_1 \cdot \cdot (m_3 - m_2)  \nonumber\\
    &+ (1 - 2\mu_0') \cdot \pi_2 (m_3 - m_2) \nonumber\\
    =& [(\mu_0 + \mu_0') \!+\! (1-2\mu_0)\cdot\pi_1 \!+\! (1-2\mu_0')\cdot\pi_2]\cdot(m_3 - m_2)\nonumber\\
    &+ m_2\nonumber\\
    =& m_2 + \mu \cdot (m_3 - m_2),
\end{align}
where $\mu = 0$ if $m_1 = 1$, otherwise, $\mu = 1$. Thus, we say $\mathcal{F}_{\mathsf{trn}}$ outputs $\lb m \rb$ correctly.
\begin{algorithm}[htbp]
	\caption{
        $\mathcal{F}_{\mathsf{trn}}(\lb m_1 \rb, \lb m_2 \rb, \lb m_3 \rb) \to \lb m \rb$}
        \vspace{2pt}
        \label{alg:trn}
	    \KwIn{$\lb m_1 \rb, \lb m_2 \rb$, and $\lb m_3 \rb$.}
        \KwOut{$\lb m \rb$, if $m_1 = 1$, then $m = m_2$, else $m = m_3$.}

        \vspace{3pt}

        \begin{mdframed}[backgroundcolor=reecolor,innerleftmargin=2pt,innerrightmargin=2pt,innerbottommargin=2pt,leftmargin=-8pt,rightmargin=12pt]
            $\triangleright$ Step 1: Computations in REE 
        \begin{mdframed}[backgroundcolor=whitecolor,innertopmargin=4pt,innerbottommargin=4pt,leftmargin=0pt,rightmargin=1pt,innerleftmargin=-4pt]
            \quad $\bullet$ $r_1 \sample \bin^{\sigma}, r_2 \sample (\frac{N}{2} - r_1, \frac{N}{2})$ and\\\hangindent = 18pt
            $r_1' \sample \bin^{\sigma}, r_2' \sample (\frac{N}{2} - r_1', \frac{N}{2})$;
            
            \quad $\bullet$ $\pi_1 \sample \bin$ and $\pi_2 \sample \bin$, \hangindent = -8pt 
            \begin{equation}
                \nonumber
	        \lb d_1 \rb \gets \begin{cases}
	            \lb m_1 \rb^{r_1} \enc(pk, r_2), &\pi_1 = 0; \\
	            \lb m_1 \rb^{-r_1} \cdot \enc(pk, r_1 + r_2), &\pi_1 = 1;	
	        \end{cases}
            \end{equation}
            \hangindent = -8pt 
            \begin{equation}
                \nonumber
	        \lb d_2 \rb \!\gets\! \begin{cases}
	            \lb m_1 \rb^{-r_1'} \cdot \enc(pk, 2r_1' + r_2'), &\pi_2 = 0; \\
	            \lb m_1 \rb^{r_1'} \cdot \enc(pk, r_2' - r_1'), &\pi_2 = 1;	
	        \end{cases}
            \end{equation}
            
            \quad $\bullet$ $\mathsf{PDec}(sk_1, \lb d_1 \rb) \to [d_1]_1$ and $\mathsf{PDec}(sk_1, \lb d_2 \rb) \to [d_2]_1$; \\ \hangindent = 18pt
            
            \quad $\bullet$ $\lb m_3 - m_2 \rb \gets \lb m_3 \rb \cdot \lb m_2 \rb^{-1}$, $\lb \pi_1 \cdot (m_3 - m_2) \rb \gets $\\ \hangindent = 19pt
            $\lb m_3 - m_2 \rb^{\pi_1} \cdot \enc(pk, 0)$, and $\lb \pi_2 \cdot (m_3 - m_2) \rb \gets \lb m_3 - m_2 \rb^{\pi_2} \cdot \enc(pk, 0)$;
        \end{mdframed}
        \end{mdframed}
        
        \vspace{-8pt}
        
        \begin{mdframed}[backgroundcolor=whitecolor,innerleftmargin=4pt,innertopmargin=0pt,innerrightmargin=2pt,innerbottommargin=-2pt,leftmargin=-8pt,rightmargin=12pt]
            REE \ {$\sendmessageright*{\langle \lb m_3 \!-\! m_2 \rb, \lb d_1 \rb, [d_1]_1, \lb \pi_1 \!\cdot\! (m_3 \!-\! m_2) \rb, \\
            \lb m_2 \rb, \lb d_2 \rb, [d_2]_1, \lb \pi_2 \!\cdot\! (m_3 \!-\! m_2) \rb \rangle}$} TEE
        \end{mdframed}
        
        \vspace{-10pt}
        
        \begin{mdframed}[backgroundcolor=teecolor,innerleftmargin=2pt,innerrightmargin=2pt,innerbottommargin=2pt,leftmargin=-8pt,rightmargin=12pt]
             $\triangleright$ Step 2: Computations in TEE
        \begin{mdframed}[backgroundcolor=whitecolor,innertopmargin=4pt,innerbottommargin=4pt,leftmargin=0pt,rightmargin=1pt,innerleftmargin=-4pt]
            \quad $\bullet$ $\mathsf{PDec}(sk_2, \lb d_1 \rb) \to [d_1]_2, \mathsf{TDec}([d_1]_1, [d_1]_2) \to d_1$ and \\\hangindent=18pt
            $\mathsf{PDec}(sk_2, \lb d_2 \rb) \to [d_2]_2, \mathsf{TDec}([d_2]_1, [d_2]_2) \to d_2$;

            \hangindent = -132pt
            \vspace{-12pt}
            \begin{flalign*}
	        \bullet \begin{cases}
	            \mu_0 = 0, &d_1 > \frac{N}{2}; \\
	            \mu_0 = 1, &d_1 < \frac{N}{2};\\
	            \mu_0' = 0, &d_2 > \frac{N}{2}; \\
	            \mu_0' = 1, &d_2 < \frac{N}{2};	
	        \end{cases}
	        \end{flalign*}

            \quad $\bullet$ $\lb R \rb \gets \lb m_2 \rb \cdot \lb m_3 - m_2 \rb^{\mu_0 + \mu_0'} \cdot \lb \pi_1 \cdot (m_3 - m_2) \rb^{1 - 2\mu_0} \cdot$ \\ \hangindent=49pt
            $ \lb \pi_2 \cdot (m_3 - m_2) \rb^{1 - 2\mu_0'} \cdot \enc(pk, 0)$;
        \end{mdframed}
        \end{mdframed}

        \vspace{-8pt}

        \begin{mdframed}[backgroundcolor=whitecolor,innerleftmargin=4pt,innertopmargin=0pt,innerrightmargin=2pt,innerbottommargin=-2pt,leftmargin=-8pt,rightmargin=12pt]
            REE \qquad $\sendmessageleft*{\lb R \rb}$ TEE
        \end{mdframed}

        \vspace{-10pt}
        
        \begin{mdframed}[backgroundcolor=reecolor,innerleftmargin=2pt,innerrightmargin=2pt,innerbottommargin=2pt,leftmargin=-8pt,rightmargin=12pt]
            $\triangleright$ Step 3: Computations in REE 
        \begin{mdframed}[backgroundcolor=whitecolor,innertopmargin=4pt,innerbottommargin=4pt,leftmargin=0pt,rightmargin=1pt,innerleftmargin=-4pt]
            \quad $\bullet$ $\mathcal{F}_{\mathsf{trn}} \gets \lb R \rb$;
        \end{mdframed}
        \end{mdframed}
        \vspace{-5pt}
\end{algorithm}

\section{\texttt{SEAT} Details}\label{SEAT Details}
In this section, we describe \texttt{SEAT}, a secure data trading solution based on our proposed TRUST. Data trading is key to driving the digital economy \cite{abla2024fair,an2023secure}. \texttt{SEAT} aims to answer the following two questions faced by existing data trading:

\textit{(1) Can the trading data be prevented from being copied?}

\textit{(2) If the answer for the first question is no, can the trading data be prevented from data resale?}

Our proposed \texttt{SEAT} argues that SOC is one of feasible methods to prevent from data resale. Roughly speaking, \texttt{SEAT} trades data use right instead of data ownership through SOC.

Fig. \ref{fig:seat} shows the system architecture of \texttt{SEAT} that consists of three entities: data seller, data broker, and data buyer. The data seller owning data sells the data use right to the data buyer with the help of the data broker. The workflow of \texttt{SEAT} involves following steps.

\ding{182} The data seller initializes system keys including a public/private key pair $(pk, sk)$ and two pairs of threshold keys $(sk_1, sk_2)$ and $(\sk_1, \sk_2)$. After that, the data seller encrypts data item by item through calling \enc, and transmits ciphertexts and corresponding keys to the data broker.

\ding{183} The data buyer submits data usage requirements to the data broker.

\ding{184} The data broker performs computations as required by the data buyer through calling SOC operations of TRUST, and then returns encrypted result $\lb \mathcal{R} \rb$ to the data seller.

\ding{185} The data seller calculates $\mathsf{PDec}(\sk_1, \lb \mathcal{R} \rb) \to [\mathcal{R}]_1$ and sends $\langle \lb \mathcal{R} \rb, [\mathcal{R}]_1, \sk_2 \rangle$ to the data buyer.

\ding{186} The data buyer calls $\mathsf{PDec}(\sk_2, \lb \mathcal{R} \rb) \to [\mathcal{R}]_2$ and $\mathsf{TDec}([\mathcal{R}]_1, [\mathcal{R}]_2)$ to obtain the final result $\mathcal{R}$.

From the workflow of \texttt{SEAT}, we see that either the data broker or the data buyer fails to learn the origin data of the data seller. In other words, no one of the data broker and the data buyer can resell the origin data from the the data seller. Although the data broker can copy the encrypted origin data, he is only able to perform operations over ciphertexts and lacks decryption capabilities, thereby being incapable to provide any data usage service.

\begin{figure}[t]
    \centering
    \includegraphics[width=\linewidth]{ 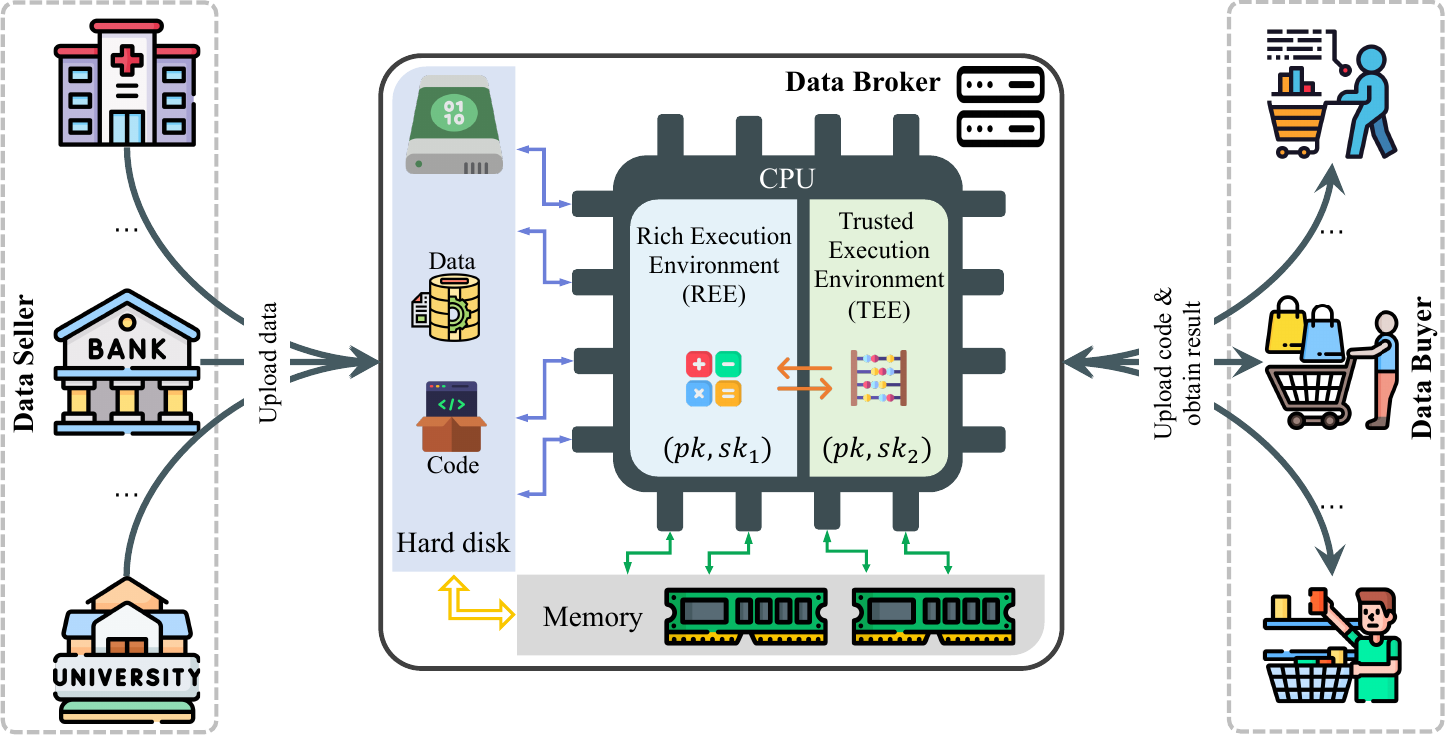}
    \caption{System architecture of \texttt{SEAT}.}
    \label{fig:seat}
\end{figure}

\section{Security Analysis}\label{Security Analysis}
In this section, we aim to demonstrate that TRUST does not leak any outsourced data to REE and TEE under the assumption of our threat model.

\begin{lemma}
$\dec(sk, \lb m \rb \cdot \enc(pk, 0)) = \dec(sk, \lb m \rb)$ is always true, and $\lb m \rb \cdot \enc(pk, 0) \neq \lb m \rb$ is also true.
\end{lemma}
\begin{proof}
We firstly prove that $\dec(sk, \lb m \rb \cdot \enc(pk, 0)) = \dec(sk, \lb m \rb)$ is true. According to the additive homomorphism of FastPaiTD, we have
\begin{equation}
    \lb m \rb \cdot \enc(pk, 0) = \lb m + 0 \rb.
\end{equation}
And since $m + 0 = m$ is always true $\forall m \in \ZZ_N$, $\dec(sk, \lb m \rb \cdot \enc(pk, 0)) = m$. Thus, $\dec(sk, \lb m \rb \cdot \enc(pk, 0)) = \dec(sk, \lb m \rb)$ is always true.

According to Eq. (\ref{eq:enc}), we have
\begin{align}
    \lb m \rb \cdot \enc(pk, 0) = &\lb m \rb \cdot (1 + N)^0 \cdot (h^r \md N)^N \nonumber\\
    =& \lb m \rb \cdot (h^r \md N)^N,
\end{align}
where $r \sample \bin^{4\kappa}$. Obviously, $(h^r \md N)^N \neq 1$, so $\lb m \rb \cdot \enc(pk, 0) \neq \lb m \rb$.
\end{proof}

\begin{lemma}
\label{lm:mr}
When $m \in (-2^\ell, 2^\ell)$, $r \sample \bin^\sigma$, and $2^{\sigma - \ell + 2}$ is a negligible function, $m + r$ is a chosen-plaintext attack secure one-time key encryption scheme, where $m$ and $r$ are the plaintext to be encrypted and the key, respectively.
\end{lemma}
\begin{proof}
Refer to the work \cite{zhao2023soci} for the detail of the proof.
\end{proof}

\begin{theorem}
\label{thrm:mul}
Given $(\lb m_1 \rb, \lb m_2 \rb)$, $\mathcal{F}_{\mathsf{mul}}$ securely computes $\lb m_1 \cdot m_2 \rb$, and does not leak $m_1$, $m_2$, or $m_1 \cdot m_2$ to REE and TEE.
\end{theorem}
\begin{proof}
At Step 1, REE only performs operations over $\lb m_1 \rb$ and $\lb m_2 \rb$. As long as FastPaiTD is secure, no $m_1$, $m_2$, or $m_1 \cdot m_2$ is leaked to REE.

At Step 2, although TEE can learn $m_1 + r$, he fails to obtain $m_1$ according to Lemma \ref{lm:mr}. Thus, we say that $m_1$ is not leaked to TEE. And since FastPaiTD is secure, TEE cannot obtain $m_2$ and $m_1 \cdot m_2$.

At Step 3, REE can learn $\lb m_2 \rb^{m_1 + r} \cdot \lb -r \cdot m_2 \rb \cdot \enc(pk, 0)$. Given $\lb m_2 \rb^{m_1 + r} \cdot \lb -r \cdot m_2 \rb$, REE is easy to learn $m_1$ under knowing $\lb m_2 \rb$, $r$, and $\lb -r \cdot m_2 \rb$. According to the computational indistinguishability experiment \cite{zhao2023soci}, REE successfully wins the experiment. Specifically, REE randomly chooses $m_0$ and $m_1$ and sends them to a challenger. After that, the challenger $b \sample \bin$ and computes and returns $\lb m_2 \rb^{m_b+r} \cdot \lb -r \cdot m_2 \rb$ to REE. In this case, if $$\lb m_2 \rb^{m_b+r} \cdot \lb -r \cdot m_2 \rb = \lb m_2 \rb^{m_0+r} \cdot \lb -r \cdot m_2 \rb,$$ REE ouputs $b = 0$, while $$\lb m_2 \rb^{m_b+r} \cdot \lb -r \cdot m_2 \rb = \lb m_2 \rb^{m_1+r} \cdot \lb -r \cdot m_2 \rb,$$ REE ouputs $b = 1$. In other words, REE always successfully guesses $b$ and win the experiment.

However, as $\lb m \rb \cdot \enc(pk, 0) \neq \lb m \rb$ (See Lemma \ref{lm:mr}), $\lb m_2 \rb^{m_b + r} \cdot \lb -r \cdot m_2 \rb \cdot \enc(pk, 0)$ may be equal to $\lb m_2 \rb^{m_0 + r} \cdot \lb -r \cdot m_2 \rb \cdot \enc(pk, 0)$ or $\lb m_2 \rb^{m_1 + r} \cdot \lb -r \cdot m_2 \rb \cdot \enc(pk, 0)$. Thus, REE fails to guess $b$ with a probability greater than $1/2$. In other words, REE cannot win the experiment, so he fails to learn $m_1$. And since FastPaiTD is secure, $m_2$ and $m_1 \cdot m_2$ are not leaked to REE.
\end{proof}

\begin{theorem}
\label{thrm:cmp}
Given $(\lb m_1 \rb, \lb m_2 \rb)$, $\mathcal{F}_{\mathsf{cmp}}$ securely compares $\lb m_1 \rb$ and $\lb m_2 \rb$, and does not leak $m_1$, $m_2$, and $\mu$ to REE and TEE.
\end{theorem}
\begin{proof}
At Step 1, REE only performs operations over $\lb m_1 \rb$ and $\lb m_2 \rb$. As long as FastPaiTD is secure, $m_1$, $m_2$, and $\mu$ are leaked to REE.

At Step 2, TEE can learn $d$ and $\mu_0$, but it is easy to demonstrate $\mu_0$ and $d$ do leak $m_1$, $m_2$, and $\mu$. Specifically, according to Algorithm \ref{alg:cmp}, we have $\mu = \mu_0 + (1 - 2\mu_0) \cdot \pi$. As FastPaiTD is secure, TEE fails to learn $\pi$ under knowing $\lb \pi \rb$. Thus, given $\mu_0$ and $\lb \pi \rb$, $\mu$ is not leaked to TEE.

Now, we adopt the computational indistinguishability experiment to prove that $d$ does not disclose $m_1$ and $m_2$. Let TEE acts as an adversary. Then, TEE randomly generates $(m_{1,0}, m_{2, 0})$ and $(m_{1,1}, m_{2, 1})$ and sends them to a challenger. Mathematically, $(m_{1,0} = m_{2,0})$ and  $(m_{1,1} = -2^\ell + 1, m_{2,1} = 2^\ell - 1)$, TEE obtains the highest probability of guessing success. After that, the challenger generates $b, \pi \sample \bin$ and $r_1 \sample \bin^\sigma, r_2 \sample (\frac{N}{2} - r_1, \frac{N}{2})$, and computes
\begin{align}
d = \left\{\begin{aligned}
&r_1 + r_2, &\text{\ for\ }b=0, \pi = 0\\
&r_1 \cdot 0 + r_2, &\text{\ for\ }b=0, \pi = 1 \\
&r_1 \cdot (-2^{\ell + 1} + 1) + r_2, &\text{\ for\ }b=1, \pi = 0\\
&r_1 \cdot (2^{\ell + 1} - 2) + r_2. &\text{\ for\ }b=1, \pi = 1 
\end{aligned}\right.\end{align}
The challenger returns $d$ to TEE. Subsequently, TEE guesses $b' = 0$ or $b' = 1$. In this case, the probability of guessing success of TEE can be formulated as
\begin{equation}
    \Pr[b'=b|d] \leq \frac{1}{2} + \frac{1}{2} \cdot \frac{1 - (2^{\ell + 1} + 1)}{2^\sigma} = \frac{1}{2} + \frac{2^\ell}{2^\sigma}.
\end{equation}
As $2^\ell \cdot 2^{-\sigma}$ is a negligible function \cite{zhao2023soci}, $\Pr[b'=b|d] \leq \frac{1}{2} + \negl[\sigma]$ holds. Thus, we say that The probability of guessing success of TEE is negligible. In other words, TEE fails to learn $m_1$ and $m_2$ from $d$.

At Step 3, REE only obtains an encrypted $\mu$, therefore, $\mu$ is not disclosed to REE.
\end{proof}

\begin{theorem}
Given $(\lb m_1 \rb, \lb m_2 \rb)$, $\mathcal{F}_{\mathsf{eql}}$ securely determines the relationship between $\lb m_1 \rb$ and $\lb m_2 \rb$, and does not leak $m_1$, $m_2$, and $\mu$ to REE and TEE.
\end{theorem}
\begin{proof}
At Steps 1 and 3, REE either performs operations over $\lb m_1 \rb$ and $\lb m_2 \rb$ or obtains an encrypted $\mu$, therefore, $m_1$, $m_2$, and $\mu$ are not leaked to REE.

At Step 2, although TEE can learn $d_1$ and $d_2$, according to Theorem \ref{thrm:cmp}, no $m_1$ or $m_2$ is disclosed to TEE. From Algorithm \ref{alg:eql}, we have $\mu = \mu_0 + (1 - 2\mu_0) \cdot \pi_1 + \mu_0' + (1 - 2\mu_0') \cdot \pi_2$. Thus, TEE cannot obtain $\mu$ as $\pi_1$ and $\pi_2$ is unknown.
\end{proof}

\begin{theorem}
Given $\lb m_1 \rb$, $\mathcal{F}_{\mathsf{abs}}$ securely computes $\lb |m_1| \rb$, and does not leak $|m_1|$ and $m_1$ to REE and TEE.
\end{theorem}
\begin{proof}
At Step 1, REE performs operations over $\lb m_1 \rb$, thus, $m_1$ is leaked to REE when FastPaiTD is secure.

At Step 2, according Lemma \ref{lm:mr} and Theorem \ref{thrm:cmp}, TEE fails to learn $\pi$ and $m_1$. And since $|m_1| = (1 - 2\mu_0) \cdot m_1 + (4\mu_0 - 2) \cdot \pi \cdot m_1$, $|m_1|$ is not disclosed to TEE due to no $\pi$ and $m_1$ being leaked.

At Step 3, REE only learns $\lb R \rb$, thus, as long as FastPaiTD is secure and Lemma \ref{lm:mr} holds, REE fails to obtain $|m_1|$.
\end{proof}

\begin{theorem}
Given $(\lb m_1 \rb, \lb m_2 \rb, \lb m_3 \rb)$, $\mathcal{F}_{\mathsf{trn}}$ securely computes $\lb m \rb$, and does not leak $m_1, m_2, m_3$ and $m$ to REE and TEE.
\end{theorem}
\begin{proof}
At Step 1, REE only performs operations over $\lb m_1 \rb$, $\lb m_2 \rb$ and $\lb m_3 \rb$, therefore, no $m_1, m_2$, or $m_3$ is leaked to REE under FastPaiTD being secure.

At Step 2, according Lemma \ref{lm:mr} and Theorem \ref{thrm:cmp}, TEE fails to learn $\pi_1,\pi_2$ and $m_1$. And since $m = m_2 + [(\mu_0 + \mu_0') + (1-2\mu_0)\cdot\pi_1 + (1-2\mu_0')\cdot\pi_2]\cdot(m_3 - m_2)$, $m$ is not leaked to TEE. Additionally, as long as FastPaiTD is secure, given $\lb m_2 \rb$ and $\lb m_3 - m_2 \rb$, no $m_2$ or $m_3$ is disclosed to TEE.

At Step 3, REE only obtains $\lb R \rb$, thus, $m$ is not leaked to REE under FastPaiTD being secure and Lemma \ref{lm:mr} holding.
\end{proof}

\begin{table*}[!t]
\centering
\caption{Comparison of Runtime and Storage Costs for Basic Cryptographic Primitive Operations (112-bit Security Level)}
\begin{threeparttable}
\rowcolors{2}{gray!25}{white} % 设置交替行颜色
\resizebox{\linewidth}{!}{ % 将表格调整为页面宽度
\begin{tabular}{c c c c c c c c c c c}
\toprule
Schemes & $\mathsf{KGen}$ & $\mathsf{Enc}$ & $\mathsf{Dec}$ & $\mathsf{PDec + TDec}$\(^1\) & Addition & Scalar-Multiplication & Subtraction & $pk$ & $sk$ & ciphertext \\
\midrule
TRUST & \textbf{61.779 ms} & \textbf{0.346 ms} & \textbf{1.617 ms} & 10.104 ms & \textbf{0.003 ms} & \textbf{0.021 ms} & \textbf{0.035 ms} & \textbf{0.500 KB} & \textbf{0.055 KB} & \textbf{0.500 KB} \\
SOCI\textsuperscript{+} \cite{zhao2024soci} & 62.334 ms & 0.347 ms & 1.619 ms & \textbf{9.401 ms} & 0.003 ms & 0.021 ms & 0.035 ms & 0.500 KB & 0.055 KB & 0.500 KB \\
SOCI \cite{zhao2023soci} & 75.776 ms & 7.076 ms & 7.072 ms & 15.005 ms & 0.003 ms & 0.021 ms & 0.035 ms & 0.500 KB & 0.250 KB & 0.500 KB \\
POCF \cite{liu2016privacy} & 74.978 ms & 7.053 ms & 7.041 ms & 14.928 ms & 0.003 ms & 0.021 ms & 0.035 ms & 0.500 KB & 0.250 KB & 0.500 KB \\
\bottomrule
\end{tabular}
}
\vspace{1mm}
\begin{tablenotes}
    \footnotesize
    \item[1] $\mathsf{PDec + TDec}$ refers to executing $\mathsf{PDec}$ twice (with keys $sk_1$ and $sk_2$, respectively) followed by one $\mathsf{TDec}$ operation.
\end{tablenotes}
\end{threeparttable}
\label{table:basic_operation_comparison}
\end{table*}
\begin{table*}[!ht]
\centering
\caption{Comparison of Computation and Communication Costs (112-bit Security Level)}
\begin{threeparttable}
\begin{tabular}{c cccc | cccc}
\toprule
\multirow{2}{*}{Protocols} & \multicolumn{4}{c|}{Computation costs (ms)} & \multicolumn{4}{c}{Communication costs (KB)} \\
\cmidrule{2-9}
 & TRUST & SOCI\textsuperscript{+}\cite{zhao2024soci} & SOCI \cite{zhao2023soci} & POCF \cite{liu2016privacy} & TRUST & SOCI\textsuperscript{+} \cite{zhao2024soci} & SOCI \cite{zhao2023soci} & POCF \cite{liu2016privacy} \\
\midrule
\textbf{$\mathcal{F}_{\mathsf{mul}}$} & \textbf{11.532} & 23.084 & 62.120 & 71.271 & \textbf{0} & 1.499 & 2.499 & 2.499 \\
\rowcolor{gray!25}
\textbf{$\mathcal{F}_{\mathsf{cmp}}$} & \textbf{11.476} & 24.941 & 46.801 & 47.791 & \textbf{0} & 1.499 & 1.499 & 1.499 \\
\textbf{$\mathcal{F}_{\mathsf{eql}}$} & \textbf{23.013} & 36.224 & 85.881 & 253.658 & \textbf{0} & 3.497 & 3.497 & 7.998 \\
\rowcolor{gray!25}
\textbf{$\mathcal{F}_{\mathsf{abs}}$} & \textbf{11.024} & 20.965 & 49.115 & 48.728 & \textbf{0} & 2.498 & 2.498 & 2.498 \\
\textbf{$\mathcal{F}_{\mathsf{trn}}$} & \textbf{22.004} & 36.228 & 70.913 & 70.322 & \textbf{0} & 4.497 & 4.497 & 4.497 \\
\bottomrule
\end{tabular}
\vspace{1mm}
\begin{tablenotes}
    \footnotesize
    \item \textbf{Note.} The computation costs is the sum of computation time and communication time. Also, we implement $\mathcal{F}_{\mathsf{eql}}$, $\mathcal{F}_{\mathsf{abs}}$, and $\mathcal{F}_{\mathsf{trn}}$ for SOCI\textsuperscript{+}, SOCI, and FOCF by using their own underlying protocols and the idea of TRUST.
\end{tablenotes}
\end{threeparttable}
\label{table:basic_protocol_comparison}
\end{table*}

\section{Experimental Evaluations}\label{Experimental Evaluation}
In this section, we extensively evaluate the proposed TRUST and \texttt{SEAT} based on TRUST. Note that all experiments are measured 1000 times and take their average as the experimental result.

\subsection{Experimental Setup}
\textbf{System Configuration}. 
To evaluate the performance of TRUST, we construct the framework using Intel\textsuperscript{\textregistered} SGX in hardware mode as a specific implementation of TEE. In detail, we implement TRUST\footnote{This work has been open-sourced and can be accessed at: \href{https://github.com/Aptx4869AC/TRUST}{https://github.com/Aptx4869AC/TRUST}} and \texttt{SEAT} in C++ and conduct experiments on a single server. The server runs a 64-bit Ubuntu 20.04 LTS operating system and is equipped with an Intel\textsuperscript{\textregistered} Xeon\textsuperscript{\textregistered} Silver 4410Y CPU @ 2.00 GHz, along with 128GB of memory, 64GB of which is allocated to the EPC.

\textbf{Network Configuration}. 
For the comparison schemes SOCI\textsuperscript{+} \cite{zhao2024soci}, SOCI \cite{zhao2023soci}, and POCF \cite{liu2016privacy}, we deploy two distinct programs on this server, each representing different entities (e.g., the cloud platform and the computation service provider) to simulate communication in a twin-server architecture within a LAN environment, with bandwidth reaching 46.4 Gbps (5.8 GB/s). In contrast, we simulate data transmission between trusted and untrusted environments by implementing secure interfaces (e.g., ECall and OCall functions provided by SGX Enclave). In our testbed, communication costs are defined as the volume of data transmitted between servers. However, in our proposed TRUST architecture, which operates on a single server, inter-server data transmission is avoided. Furthermore, the communication latency between trusted and untrusted regions remains minimal due to intra-server data transfer.
 
Note that TRUST also adopts the offline-online mechanism proposed by Zhao \textit{et al}. \cite{zhao2024soci} to accelerate the computations. In this work, REE and TEE environments independently generate tuples during the offline phase. Specifically, the encryptions of random numbers and certain constants are precomputed offline, while the online phase handles only remaining operations from the offline phase. For further details, refer to \cite{zhao2024soci}.

\subsection{Basic Cryptographic Primitive Operations Evaluations}
In this subsection, we evaluate the performance of basic cryptographic primitive operations (including $\mathsf{KGen}$, $\mathsf{Enc}$, $\mathsf{Dec}$, $\mathsf{PDec + TDec}$, and basic homomorphic operations) for TRUST and compare it with SOCI\textsuperscript{+} \cite{zhao2024soci}, SOCI \cite{zhao2023soci} and POCF \cite{liu2016privacy}. The comparison of runtime and storage costs for basic cryptographic primitive operations between different schemes are presented in Table \ref{table:basic_operation_comparison}, where the bit length of $N$ is 2048 to achieve 112-bit security. 

From Table \ref{table:basic_operation_comparison}, we see that TRUST outperforms SOCI and POCF in runtime and storage costs. Additionally, TRUST provides a similar performance with SOCI\textsuperscript{+} for basic cryptographic primitive operations.

\begin{figure}[]
    \centering
    \subfloat[Computation cost of $\mathcal{F}_{\mathsf{mul}}$]
    {
        \includegraphics[width=1.625in,height=1.4in]
        { 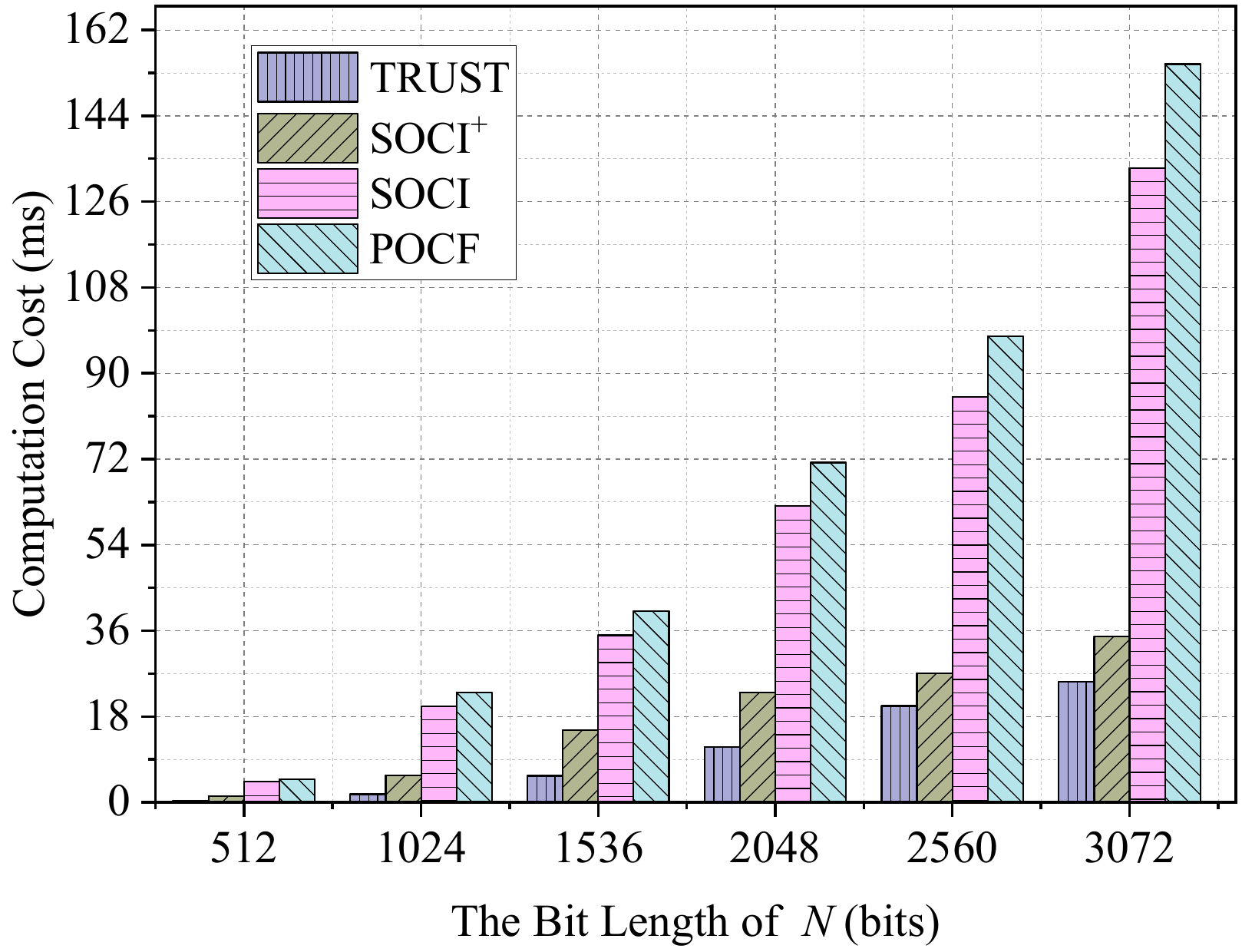}
    }
    \subfloat[Communication cost of $\mathcal{F}_{\mathsf{mul}}$]
    {
        \includegraphics[width=1.625in,height=1.4in]
        { 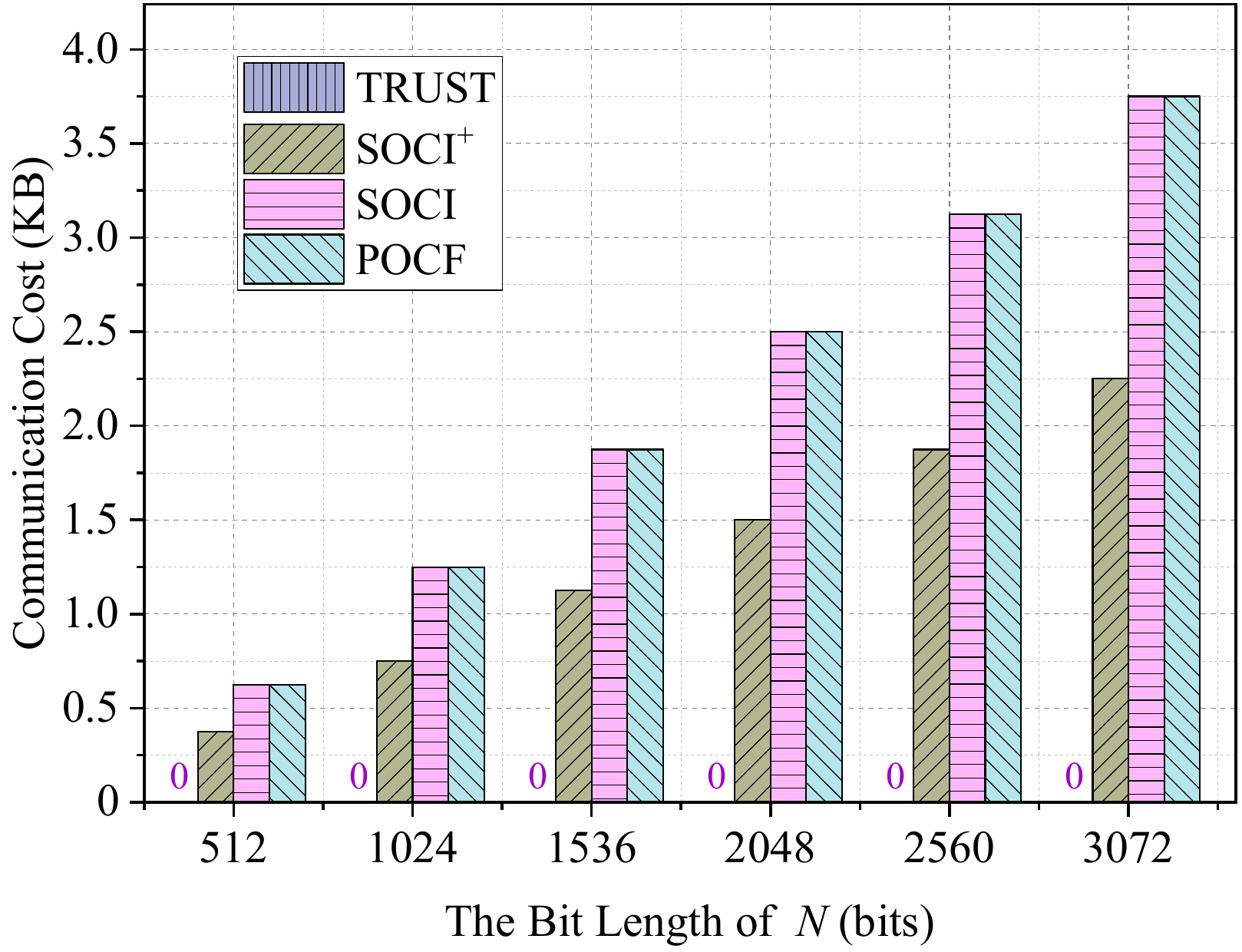}
    }
    \par\smallskip
    
    \subfloat[Computation cost of $\mathcal{F}_{\mathsf{cmp}}$]
    {
        \includegraphics[width=1.625in,height=1.4in]
        { 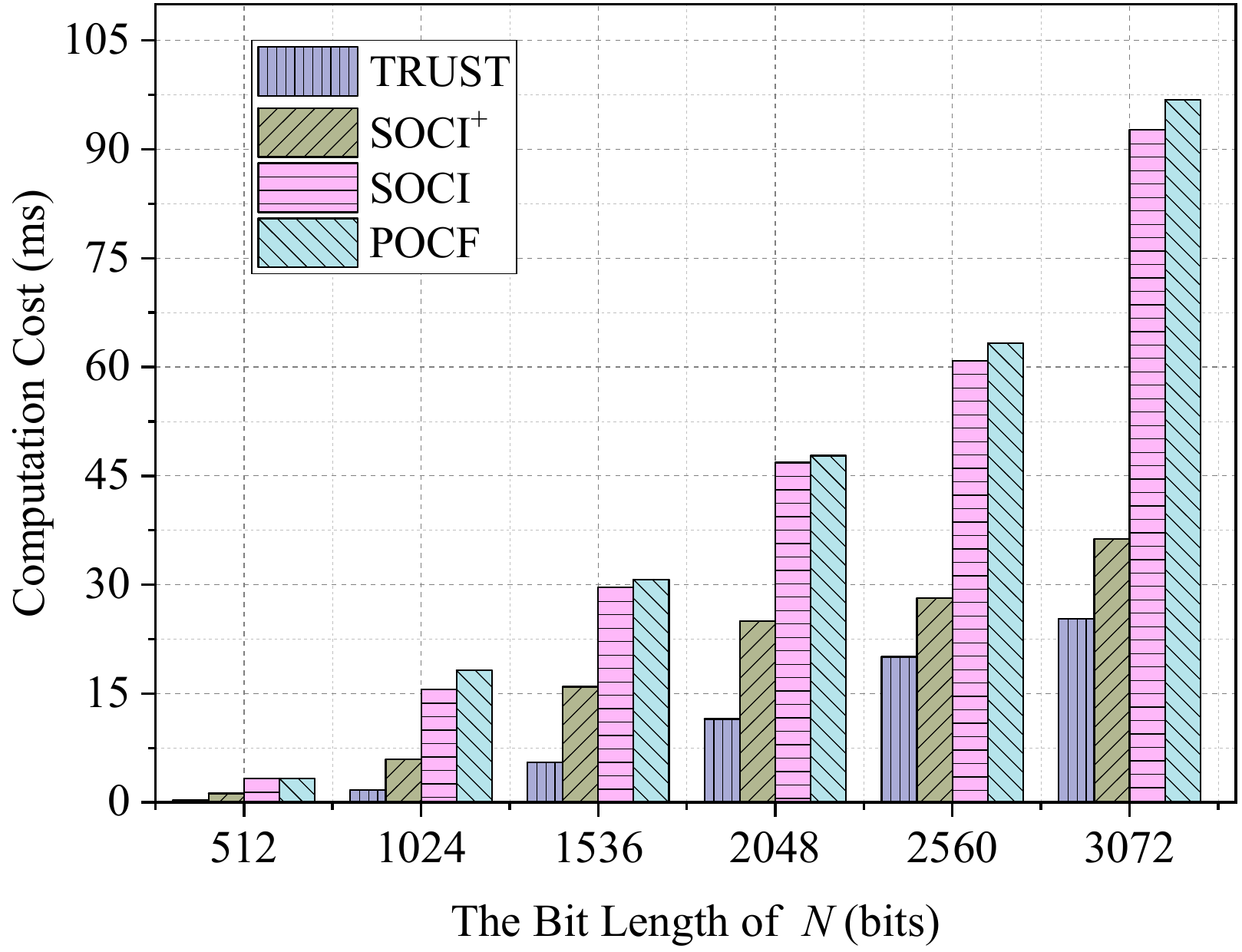}
    }
    \subfloat[Communication cost of $\mathcal{F}_{\mathsf{cmp}}$]
    {
        \includegraphics[width=1.625in,height=1.4in]
        { 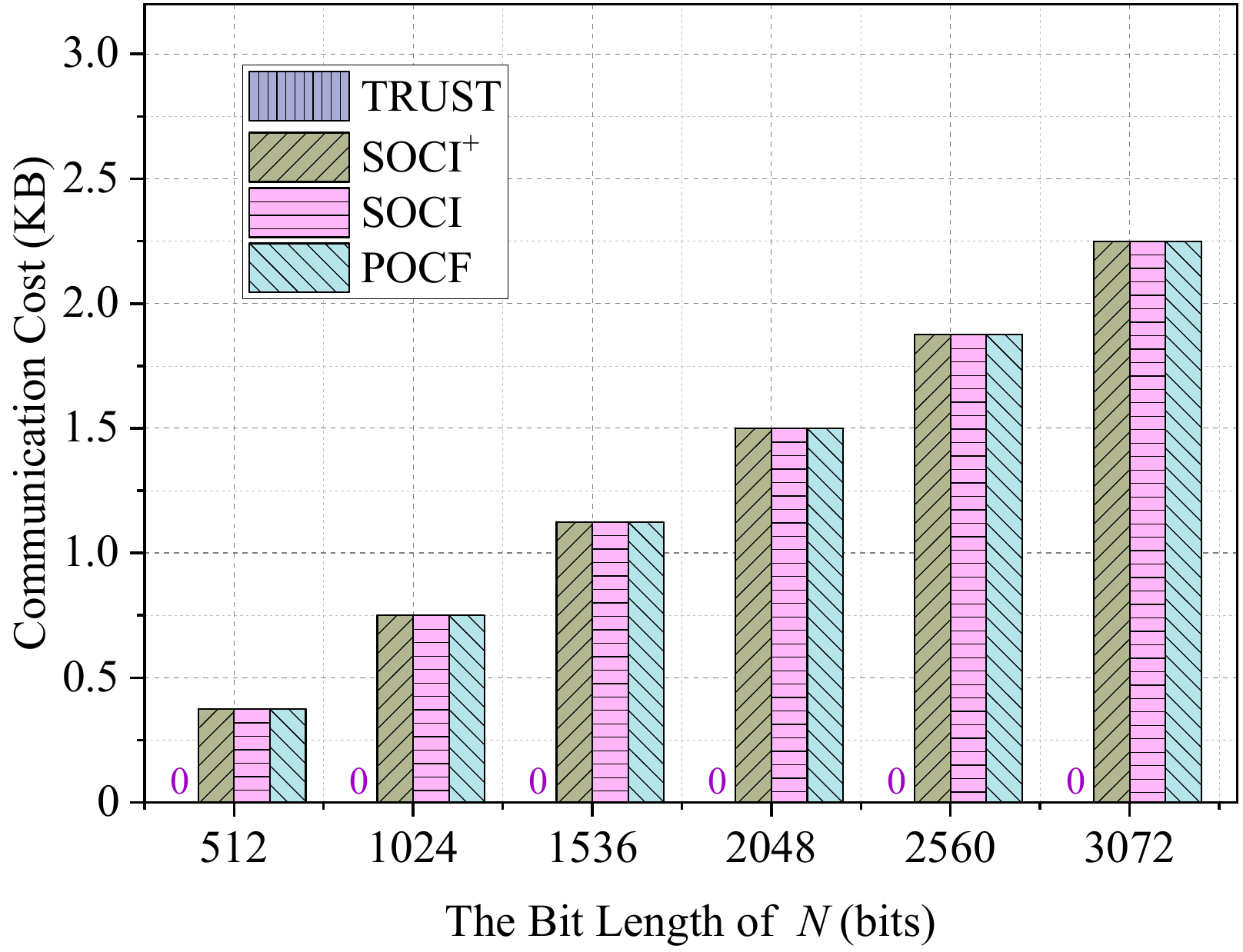}
    }
    \par\smallskip
    
    \subfloat[Computation cost of $\mathcal{F}_{\mathsf{eql}}$]
    {
        \includegraphics[width=1.625in,height=1.4in]
        { 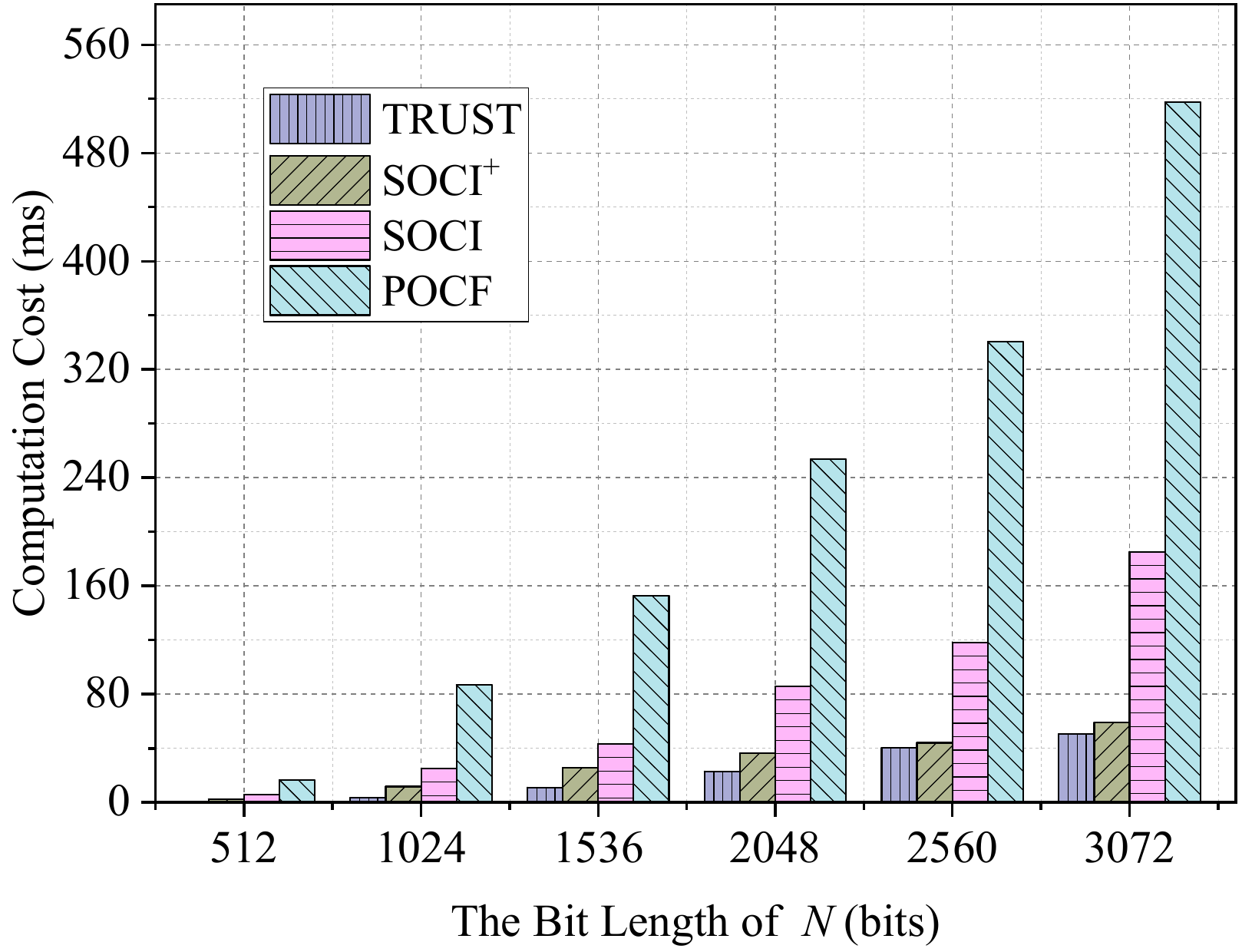}
    }
    \subfloat[Communication cost of $\mathcal{F}_{\mathsf{eql}}$]
    {
        \includegraphics[width=1.625in,height=1.4in]
        { 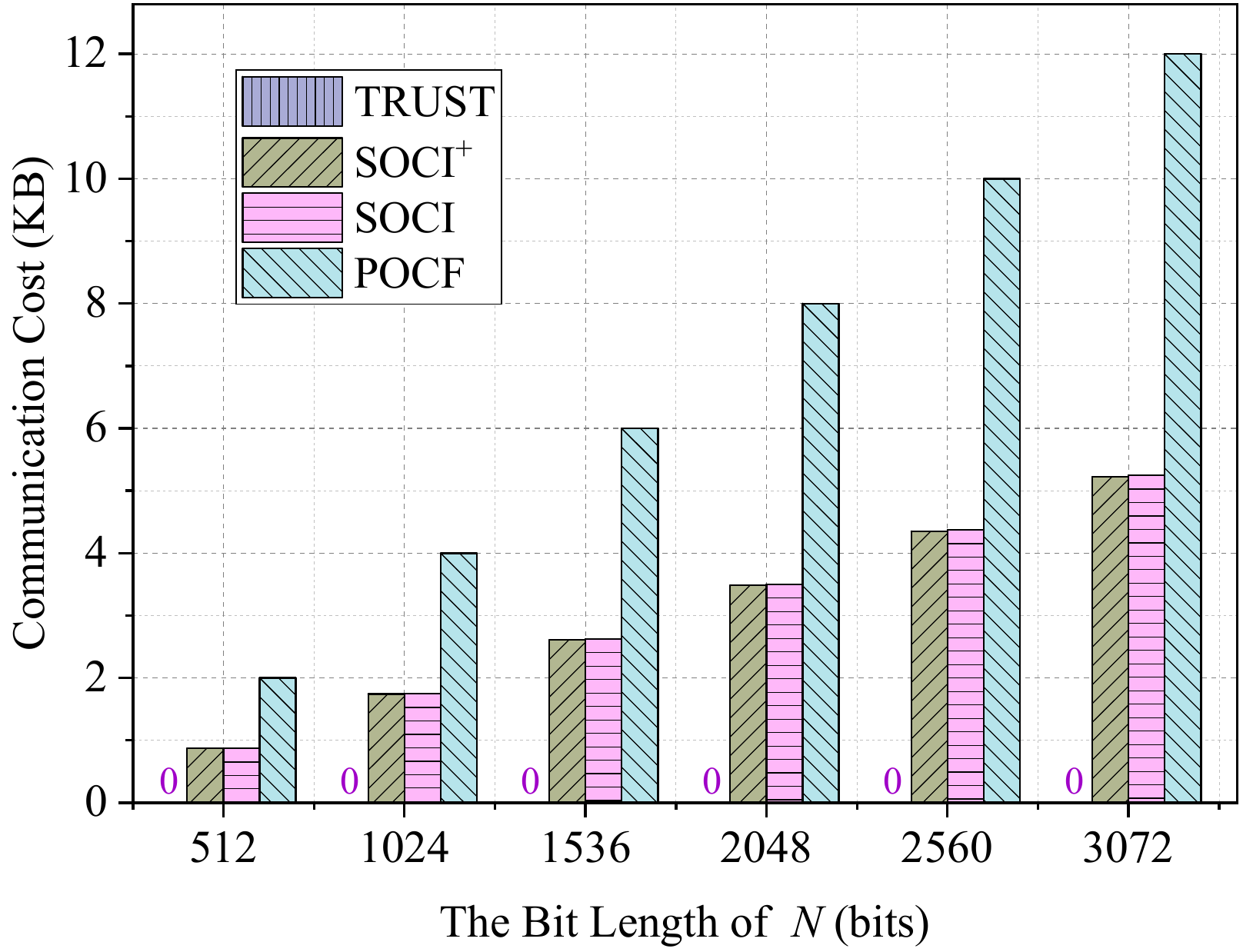}
    }

    \par\smallskip
    \subfloat[Computation cost of $\mathcal{F}_{\mathsf{abs}}$]
    {
        \includegraphics[width=1.625in,height=1.4in]
        { 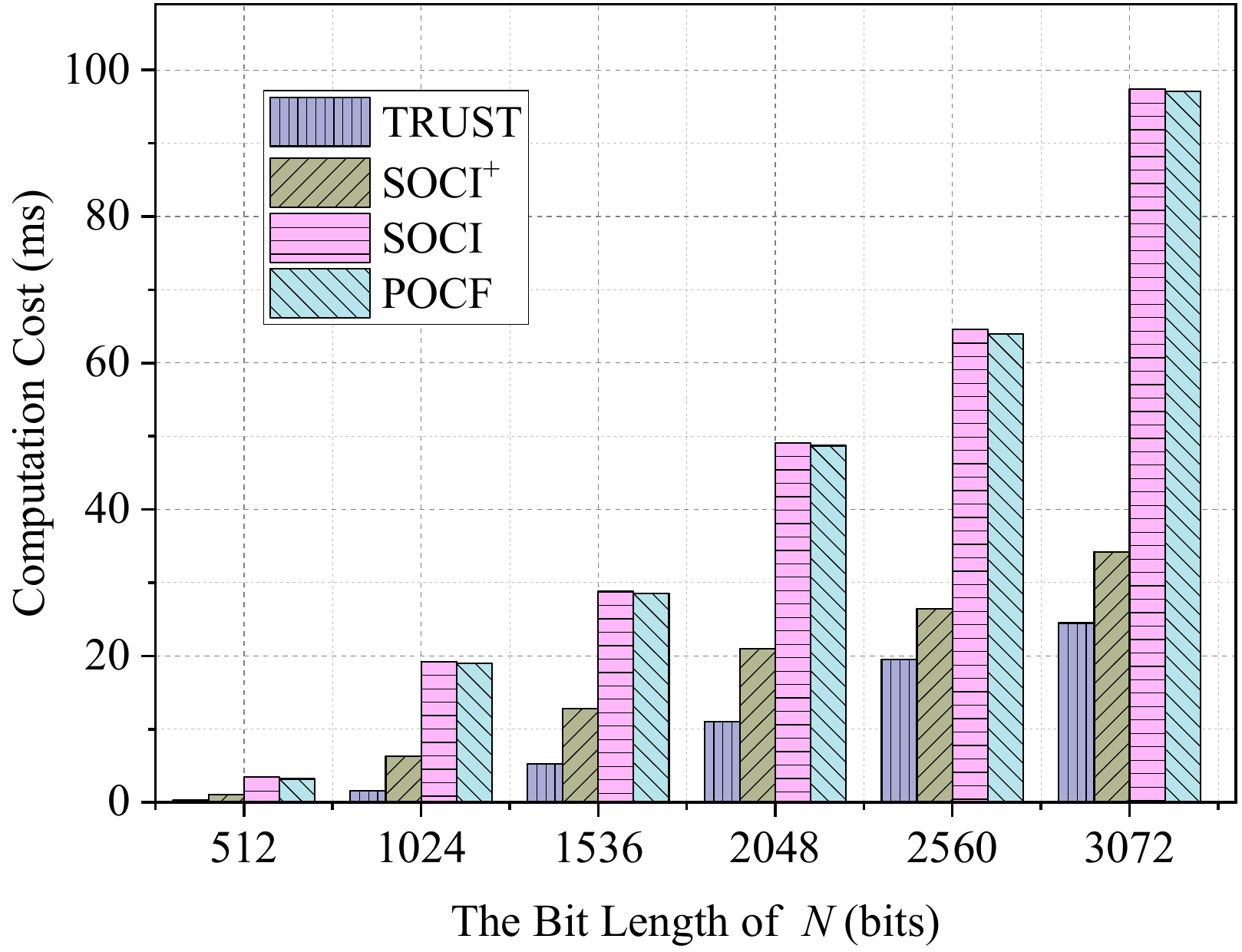}
    }
    \subfloat[Communication cost of $\mathcal{F}_{\mathsf{abs}}$]
    {
        \includegraphics[width=1.625in,height=1.4in]
        { 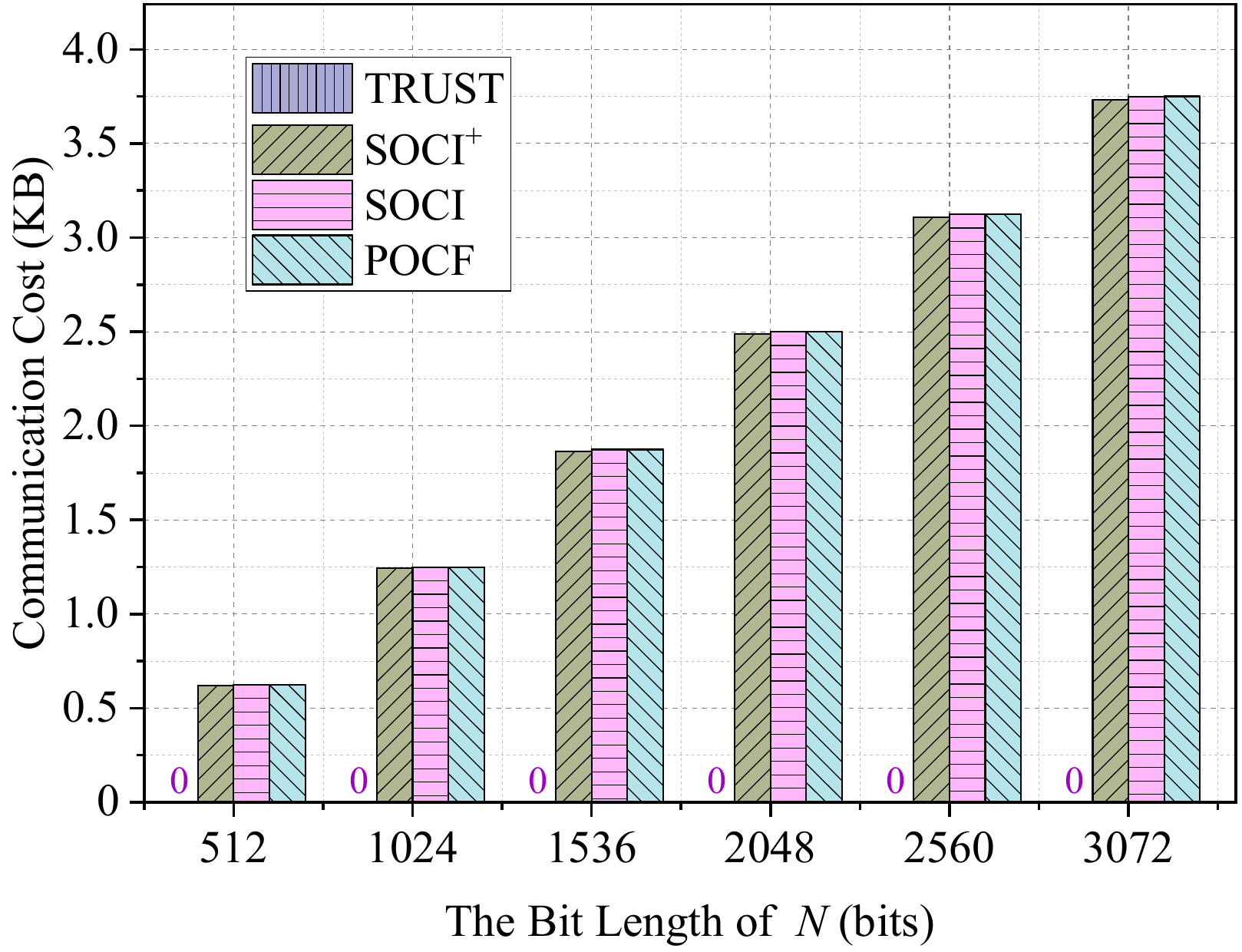}
    }

    \par\smallskip
    
    \subfloat[Computation cost of $\mathcal{F}_{\mathsf{trn}}$]
    {
        \includegraphics[width=1.625in,height=1.4in]
        { 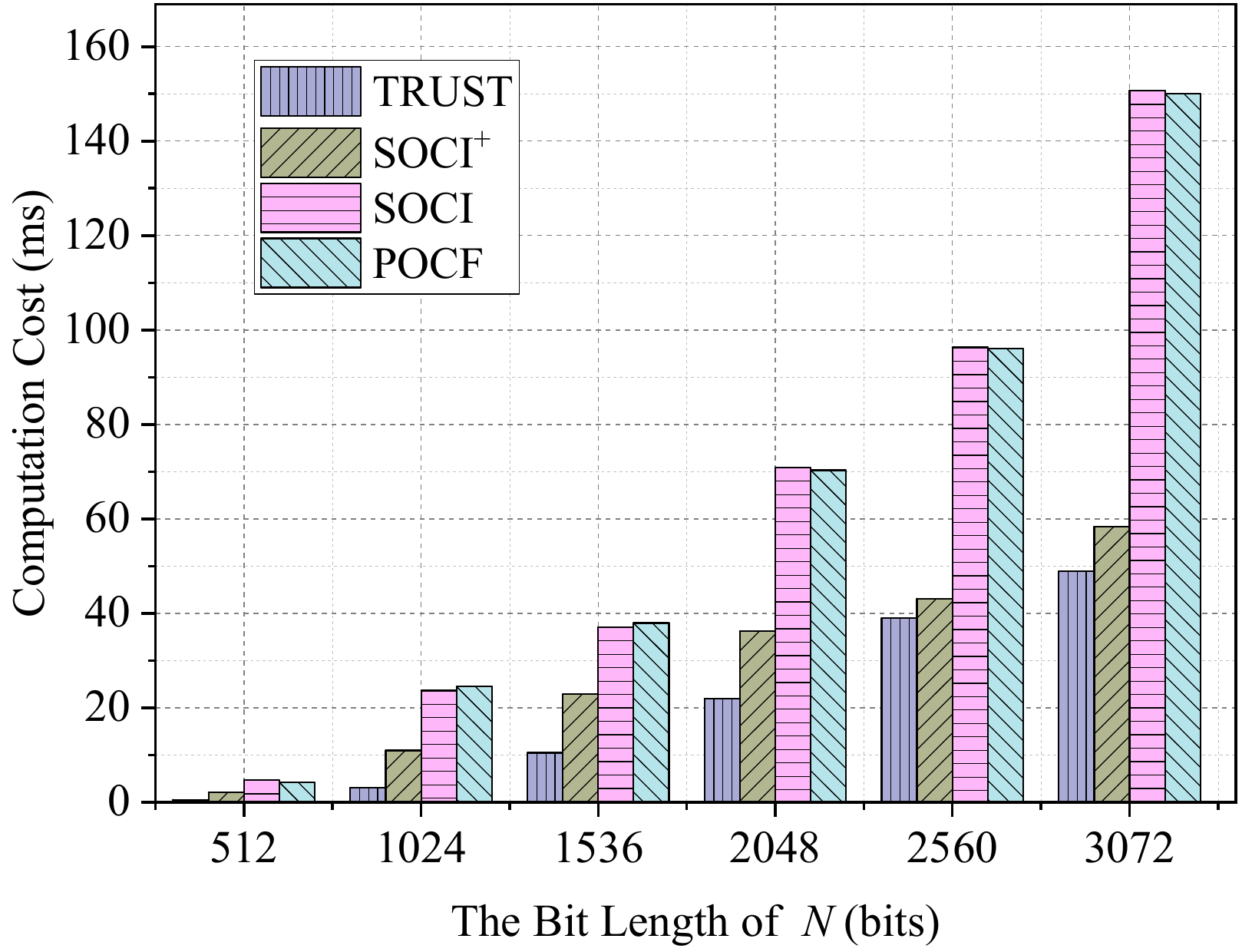}
    }
    \subfloat[Communication cost of $\mathcal{F}_{\mathsf{trn}}$]
    {
        \includegraphics[width=1.625in,height=1.4in]
        { 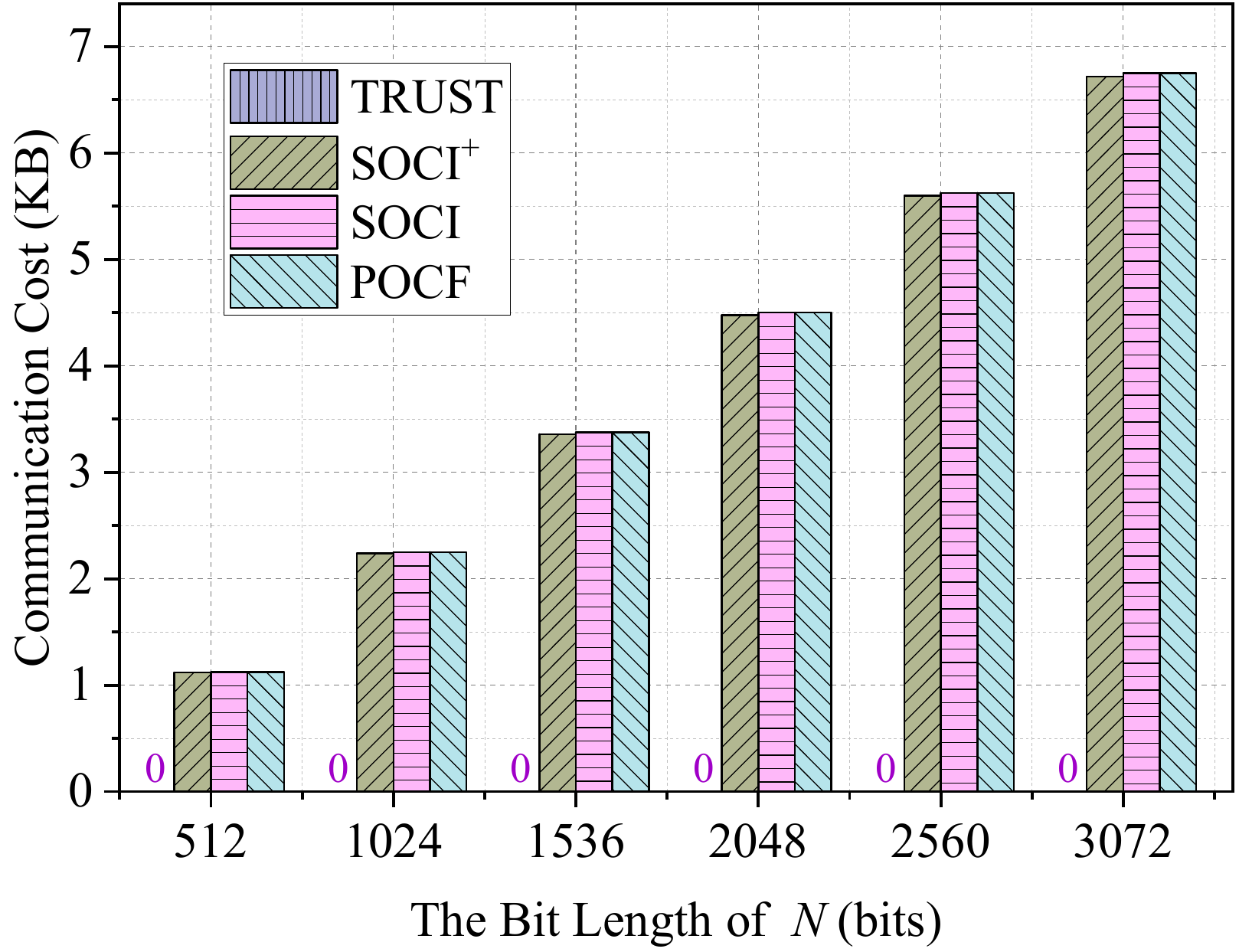}
    }
    \caption{Comparison of computation and communication costs for different schemes with a varying bit-length $N$ ($\ell = 32, \sigma = 128$).}
    \label{fig:basic_protocol_comparison}
\end{figure}

\subsection{Protocol Evaluations}
Table \ref{table:basic_protocol_comparison} presents the computation and communication costs of TRUST, SOCI\textsuperscript{+}, SOCI, and POCF at the 112-bit security level. The experimental results show that TRUST outperforms SOCI\textsuperscript{+}, SOCI, and POCF in both computation and communication costs. Specifically, experimental results indicate that TRUST improves computation costs by 1.6 - 2.1 times compared to SOCI\textsuperscript{+}, which is the state-of-the-art in secure multiplication and secure comparison protocols within the twin-server architecture based on FastPaiTD. One possible explanation is that TRUST designs a more efficient multiplication protocol and presents a more concise system architecture.

Additionally, as shown in Table \ref{table:basic_protocol_comparison}, TRUST demonstrates significant advantages in communication costs over SOCI\textsuperscript{+}, SOCI, and POCF across all protocols. The reason for this is that our experiments simulate real network communication within a LAN, where frequent serialization and deserialization during data transfer between different servers introduce additional performance overhead. This leads to increased communication latency, higher bandwidth consumption, and potential network congestion, all of which negatively affect overall system throughput and the efficiency of large-scale computations.
Thanks to our single-server architecture for TRUST, data transfer occurs within the same machine or with minimal context switching, eliminating these issues. 

As depicted in Fig. \ref{fig:basic_protocol_comparison}, we give the computation and communication costs of SOC protocols for TRUST, SOCI\textsuperscript{+}, SOCI, and POCF with a varying bit-lengths of $N$, providing an intuitive comparison. Fig. \ref{fig:basic_protocol_comparison} demonstrates that TRUST exhibits significant advantages, particularly for critical protocols such as $\mathcal{F}_{\mathsf{mul}}$ and $\mathcal{F}_{\mathsf{cmp}}$, with both requiring only around 12 ms at $|N| = 2048$. Moreover, the communication costs across all protocols in TRUST remain unaffected as the parameter $N$ increases, consistently showing as 0. Arguably, TRUST outperforms related solutions in both computation and communication costs, suggesting that TRUST would offer significant advantages in high-performance computing, and secure data processing applications.

\subsection{\texttt{SEAT} Evaluations}
To evaluate \texttt{SEAT}, we assume the data buyer submits a requirement about training a linear regression (LR) model. The data seller holds the raw data, including the feature matrix $\mathbf{X}$ and the target vector $\vec{y}$. And the data buyer holds the initial model, which includes the weights $w$, the bias $b$, and the learning rate $\eta$, all specific to LR. Finally, the data broker, which is a single TEE-equipped cloud server, is responsible for securely outsourcing the computation based on the requirements provided by the data buyer.

To show the effectiveness, we compare \texttt{SEAT} with SDTE \cite{dai2019sdte} using AES-GCM (Key: 256-bit; IV: 96-bit; Tag: 128-bit) and the \texttt{Baseline} without any privacy preservation over two open-source datasets: Student Performance\footnote{This dataset is available at: \href{https://www.kaggle.com/datasets/nikhil7280/student-performance-multiple-linear-regression}{https://www.kaggle.com/datasets/nikhil7280/student-performance-multiple-linear-regression}} and Fish Market\footnote{This dataset is available at: \href{https://www.kaggle.com/datasets/vipullrathod/fish-market}{https://www.kaggle.com/datasets/vipullrathod/fish-market}}. 

As shown in Table \ref{table:PPLinR}, \texttt{SEAT} generates the same result as the \texttt{Baseline} and SDTE under three common evaluation metrics. Fig. \ref{fig:LR_comparison} intuitively shows the performance during training LR between the \texttt{Baseline} and \texttt{SEAT}. In addition, from Fig. \ref{fig:LR_comparison}, we see that \texttt{SEAT} performs equally as the \texttt{Baseline} in terms of training. According to the results of Table \ref{table:PPLinR} and Fig. \ref{fig:LR_comparison}, we argue that \texttt{SEAT} is feasible.

\begin{table*}[!ht]
    \centering
    \caption{Cross-Dataset Feasibility Study and Performance Assessment of Baseline, SDTE, and SEAT}
    \begin{threeparttable}
    \resizebox{\linewidth}{!}{ % 将表格调整为页面宽度
    \begin{tabular}{c cc|c |ccc| ccc| ccc}
        \toprule
        \multirow{2}{*}{Dataset} & \multirow{2}{*}{Scale} & \multirow{2}{*}{Settings} & \multirow{2}{*}{Iterations} & \multicolumn{3}{c}{\texttt{Baseline}} & \multicolumn{3}{c}{SDTE \cite{dai2019sdte}} & \multicolumn{3}{c}{\texttt{SEAT}} \\
        \cmidrule(lr){5-7} \cmidrule(lr){8-10} \cmidrule(lr){11-13}
        & & & & MSE & $R^2$ & MAE & MSE & $R^2$ & MAE & MSE & $R^2$ & MAE \\
        \midrule
        \multirow{4}{*}{Student Performance} & \multirow{4}{*}{$10000 \times 6$} & \multirow{4}{*}{
            \begin{tabular}{@{}l@{}} 
                learning rate ($\eta$): 0.0001 \\
                bias ($b$): 0.0 \\ 
                weights ($w$): 5 dimensions, all 0 \\ 
                batch size: 64 \\
            \end{tabular}
        }
        & 
        200 & 30.072 & 0.917 & 4.439 & 30.072 & 0.917 & 4.439 & 30.072 & 0.917 & 4.439 \\
        & & & 
        % \rowcolor{gray!25}
        400 & 26.167 & 0.928 & 4.135 & 26.167 & 0.928 & 4.135 & 26.167 & 0.928 & 4.135  \\
        & & & 
        600 & 22.904 & 0.937 & 3.866 & 22.904 & 0.937 & 3.866 & 22.904 & 0.937 & 3.866 \\
        & & & 
        % \rowcolor{gray!25}
        800 & 20.140 & 0.944 & 3.622 & 20.140 & 0.944 & 3.622 & 20.140 & 0.944 & 3.622 \\
        \midrule
        \multirow{4}{*}{Fish Market} & \multirow{4}{*}{$159 \times 7$} & \multirow{4}{*}{
            \begin{tabular}{@{}l@{}} 
                learning rate ($\eta$): 0.0001 \\
                bias ($b$): 0.0 \\ 
                weights ($w$): 6 dimensions, all 0 \\ 
                batch size: 16 \\
            \end{tabular}
        }
        & 2000 & 28946.865 & 0.721 & 141.262 & 28946.865 & 0.721 & 141.262 & 28946.865 & 0.721 & 141.262 \\
        & & & 
        % \rowcolor{gray!25}
        4000 & 27153.971 & 0.738 & 137.225 & 27153.971 & 0.738 & 137.225 & 27153.971 & 0.738 & 137.225 \\
        & & & 6000 & 26344.280 & 0.746 & 134.947 & 26344.280 & 0.746 & 134.947 & 26344.280 & 0.746 & 134.947 \\
        & & & 
        % \rowcolor{gray!25}
        8000 & 25907.904 & 0.750 & 133.422 & 25907.904 & 0.750 & 133.422 & 25907.904 & 0.750 & 133.422 \\
        \bottomrule
    \end{tabular}
    }
    \vspace{1mm}
    \begin{tablenotes}
    \footnotesize
     \item \textbf{Remark.} The training and test sets each comprise half of their total sample count. $\text{MSE} = \frac{1}{n} \sum_{i=1}^{n} (y_i - \hat{y}_i)^2$, $R^2 = 1 - \frac{\sum_{i=1}^{n} (y_i - \hat{y}_i)^2}{\sum_{i=1}^{n} (y_i - \bar{y})^2}$, \\ $\text{MAE} = \frac{1}{n} \sum_{i=1}^{n} |y_i - \hat{y}_i|$ \, ($y_i$: True Value, $\hat{y}_i$: Predicted Value, $\bar{y}$: Mean of True Values).
    \end{tablenotes}
    \end{threeparttable}
    \label{table:PPLinR}
\end{table*}

\begin{figure*}[!ht]
    \centering
    \subfloat[Model training on Student Performance]{
        \includegraphics[width=0.48\linewidth]{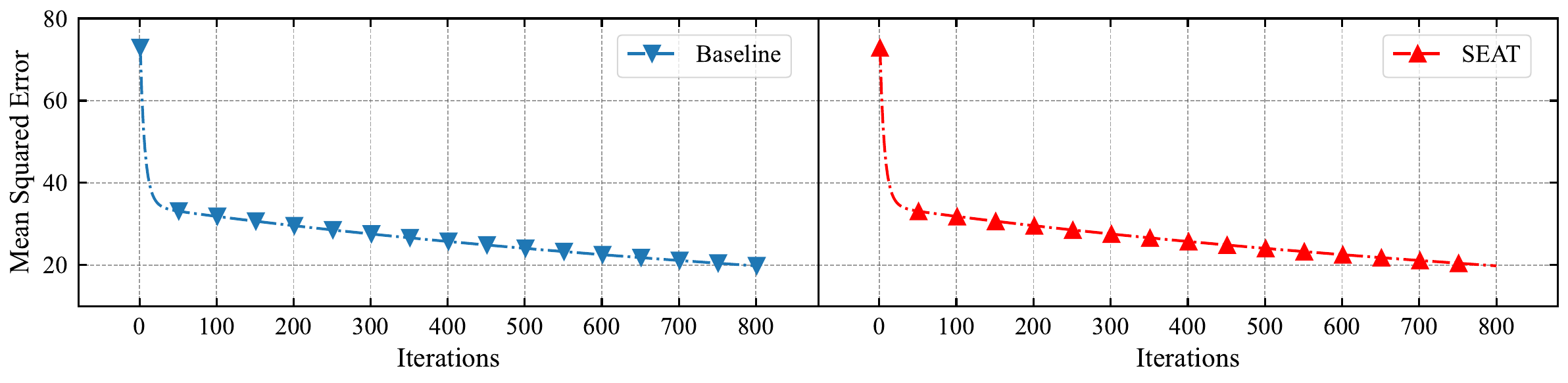} 
    }
    \subfloat[Model training on Fish Market]{
        \includegraphics[width=0.48\linewidth]{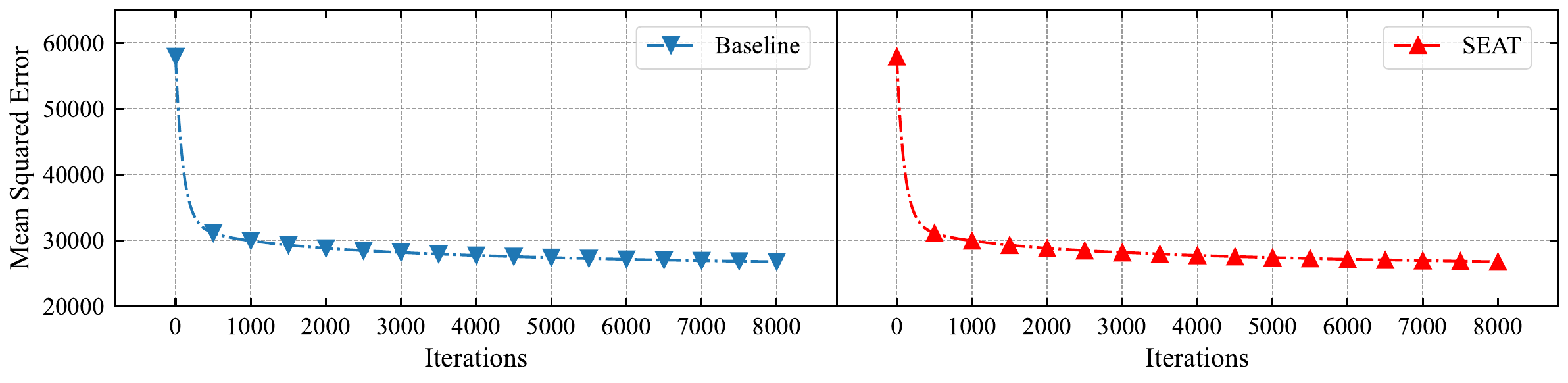} 
    }
    \caption{Gradient descent performance of \texttt{Baseline} and \texttt{SEAT} across iterations.}
    \label{fig:LR_comparison}
\end{figure*}

\section{Conclusion}\label{Conclusion}
In this work, we proposed TRUST, a toolkit for TEE-assisted secure
outsourced computation over integers by seamlessly integrating a (2, 2)-threshold Paillier cryptosystem and TEE. TRUST avoids any collusion attacks from a twin-server architecture of SOC. TRUST not only enriches computational operations, but also improves their performance. Also, we designed a secure data trading solution \texttt{SEAT} based on the proposed TRSUT and confirmed its feasibility. For future work, we will explore TRUST to more complex computation tasks, such as privacy-preserving machine learning inference.

\bibliographystyle{IEEEtran}
\bibliography{IEEEabrv,TRUST}
\begin{IEEEbiography}
[{\includegraphics[width=1in,height=1.25in,clip,keepaspectratio]{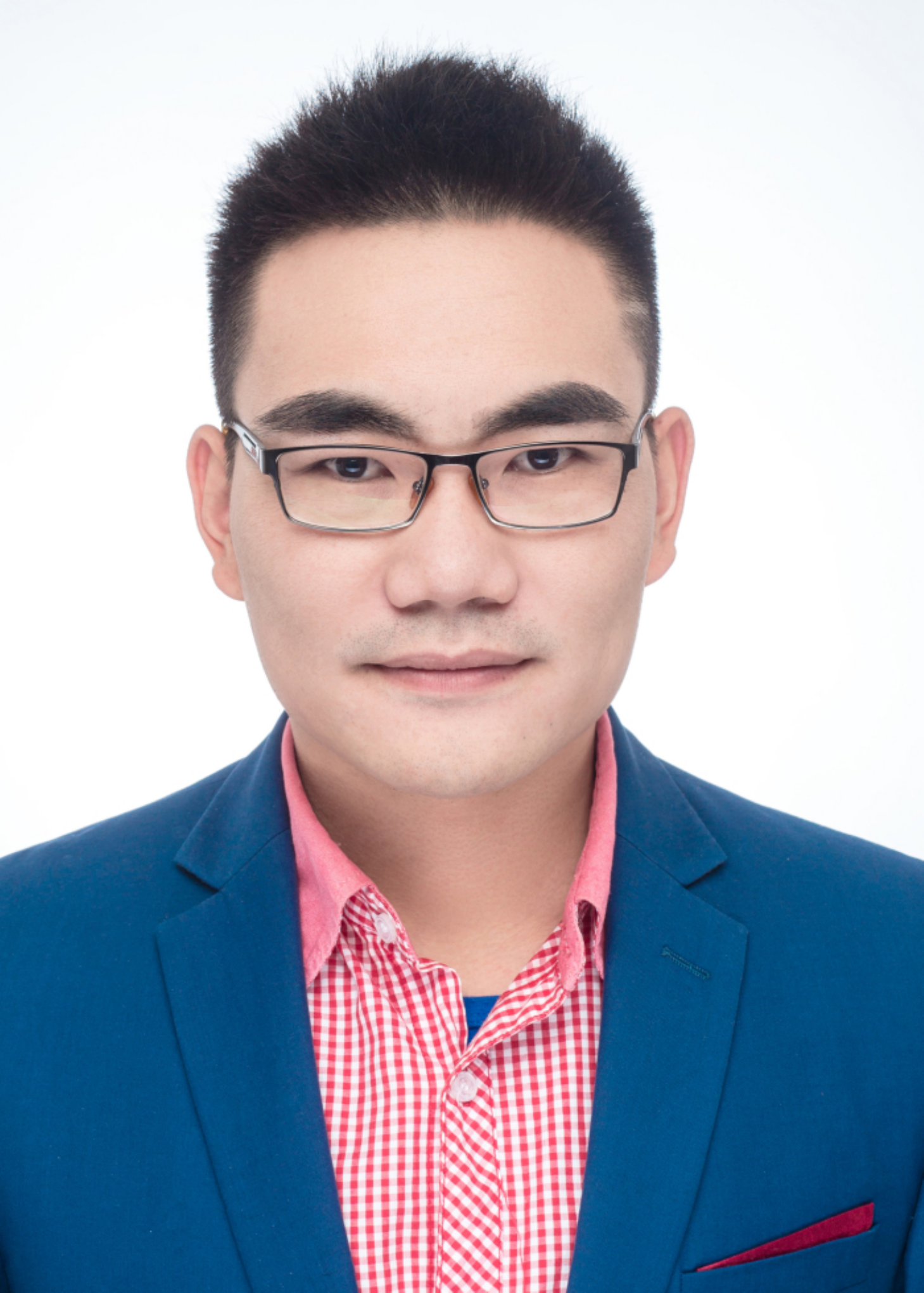}}]{Bowen Zhao}
 (Member, IEEE) received the Ph.D. degree in cyberspace security from the South China University of Technology, China, in 2020. He was a Research Scientist with the School of Computing and Information Systems, Singapore Management University, from 2020 to 2021. He is currently an Associate Professor with the Guangzhou Institute of Technology, Xidian University, Guangzhou, China. His current research interests include privacy-preserving computation and learning and privacy-preserving crowdsensing.
\end{IEEEbiography}

\begin{IEEEbiography}
[{\includegraphics[width=1in,height=1.25in,clip,keepaspectratio]{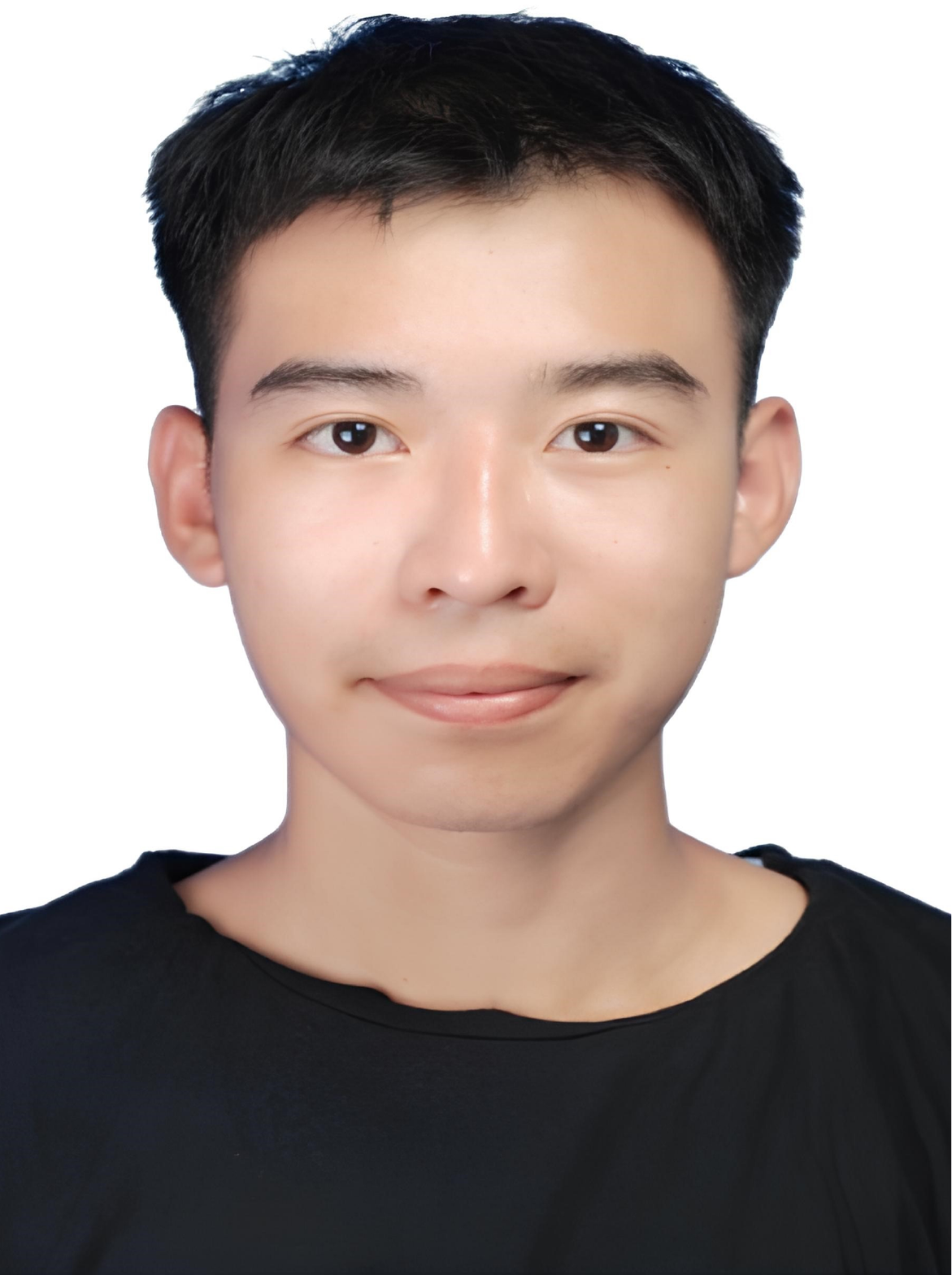}}]{Jiuhui Li}
received the B.S. degree in computer science and technology from Guangzhou University, Guangzhou, China, in 2023. He is currently working toward the M.S. degree in Guangzhou Institute of Technology, Xidian University, Guangzhou. His research interests are trusted computing and secure computation.
\end{IEEEbiography}

\begin{IEEEbiography}
[{\includegraphics[width=1in,height=1.25in,clip,keepaspectratio]{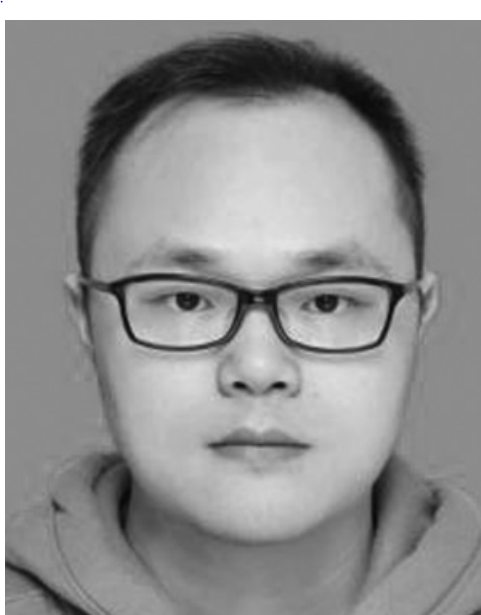}}]{Peiming Xu}
received the B.S. and Ph.D. degrees
from South China University of Technology,
Guangzhou, China, in 2014 and 2019, respectively.
He is currently a Researcher with Electric Power
Research Institute, China Southern Power Grid
Company Ltd., Guangzhou. His research interests
mainly focus on cyber security of new power system,
privacy computing, and post-quantum cryptography.
\end{IEEEbiography}

\begin{IEEEbiography}
[{\includegraphics[width=1in,height=1.25in,clip,keepaspectratio]{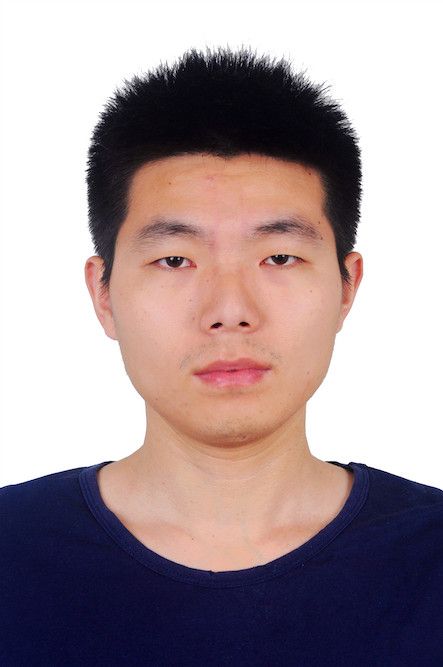}}]{Xiaoguo Li}
received his Ph.D. degree in computer science from Chongqing University, China, in 2019. He was a Research Fellow at Hong Kong Baptist University and Singapore Management University. He is currently an associate professor at Chongqing University, China. His current research interests include trusted computing, secure computation, and public-key cryptography.
\end{IEEEbiography}

\begin{IEEEbiography}
[{\includegraphics[width=1in,height=1.25in,clip,keepaspectratio]{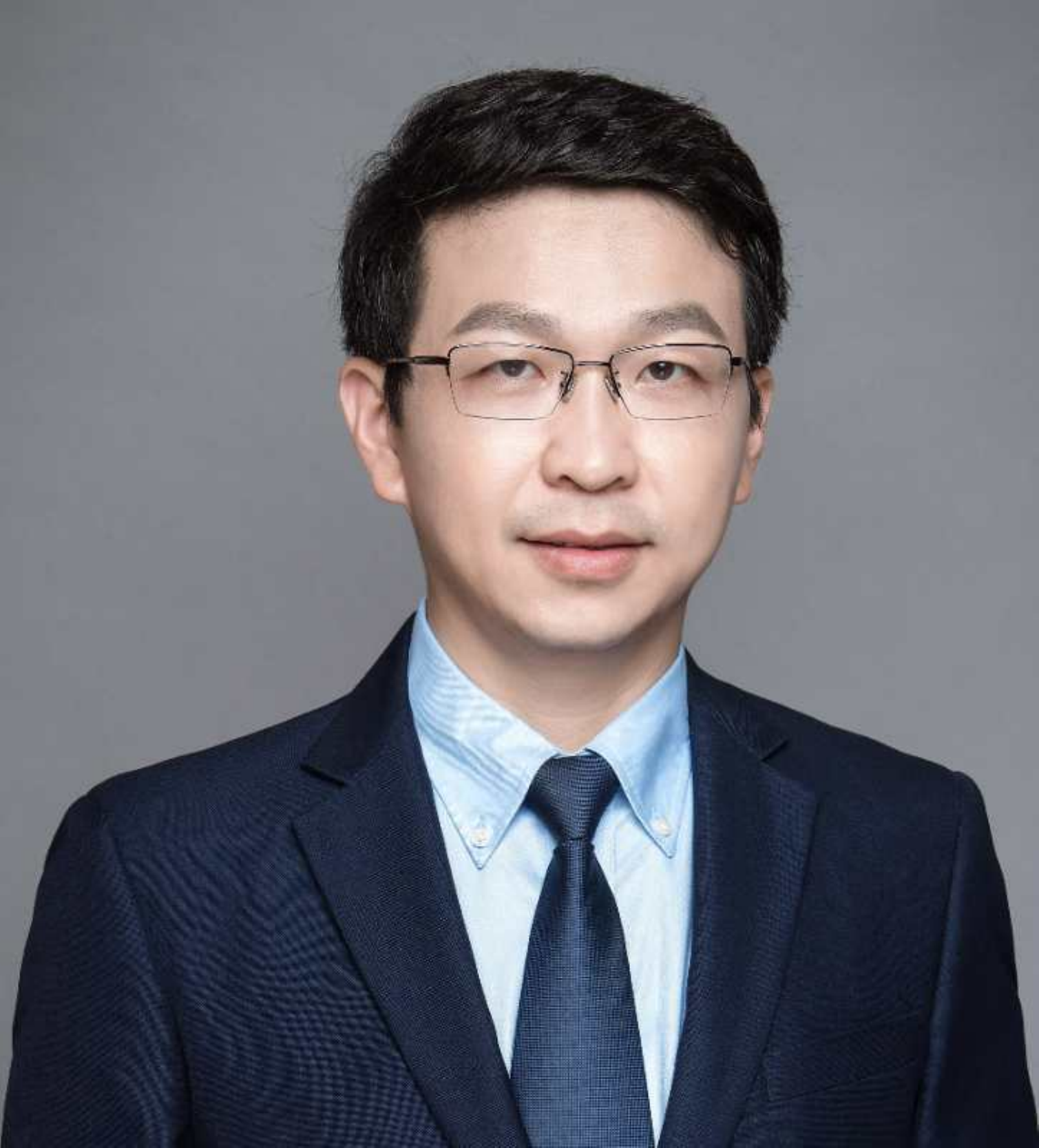}}]{Qingqi Pei}
 (Senior Member, IEEE) received the B.S., M.S., and Ph.D. degrees in computer science and cryptography from Xidian University, in 1998, 2005, and 2008, respectively. He is currently a Professor and a member of the State Key Laboratory of Integrated Services Networks, also a Professional Member of ACM, and a Senior Member of the Chinese Institute of Electronics and China Computer Federation. His research interests focus on digital content protection and wireless networks and security.
\end{IEEEbiography}

\begin{IEEEbiography}
[{\includegraphics[width=1in,height=1.25in,clip,keepaspectratio]{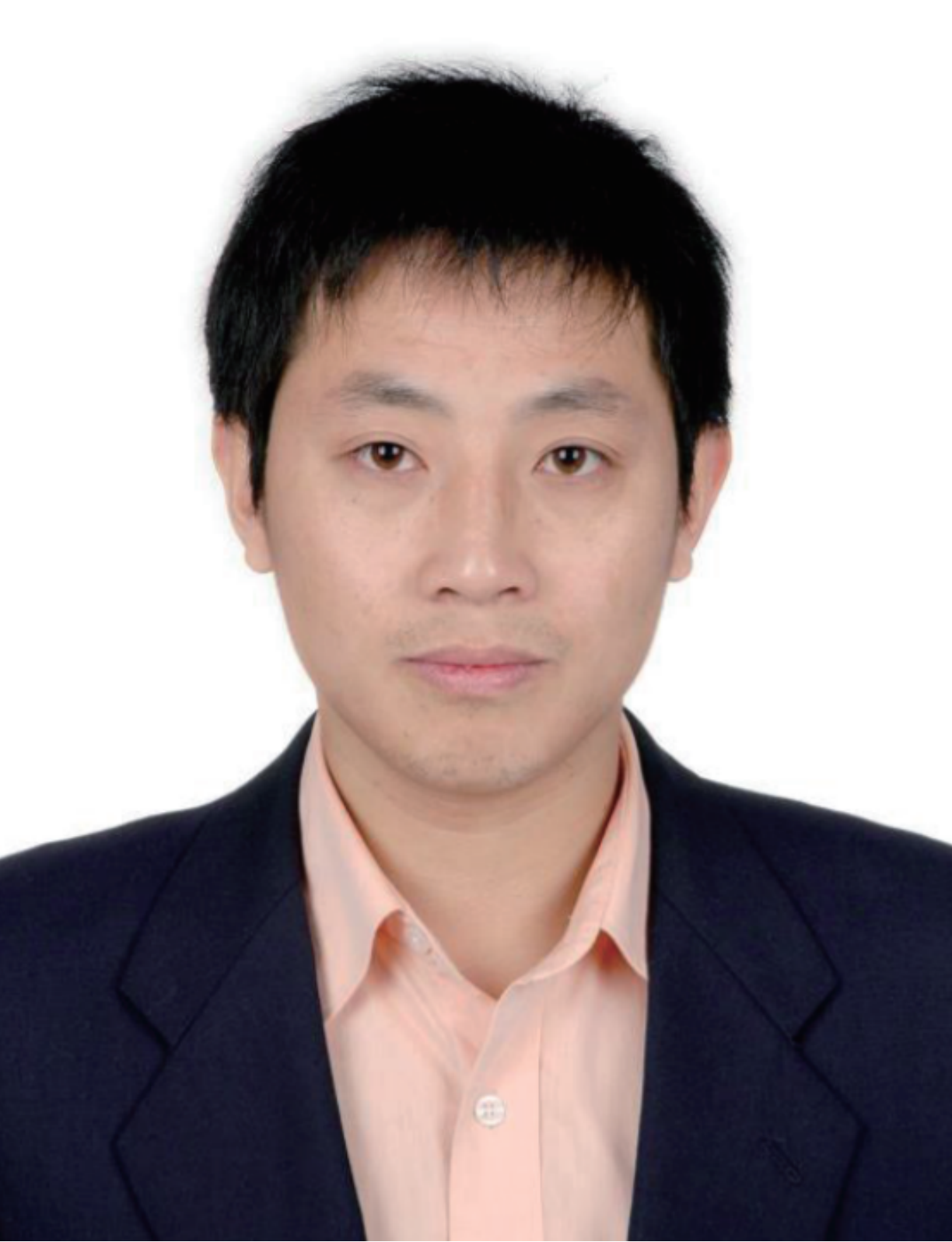}}]{Yulong Shen}
(Member, IEEE) received the B.S. and M.S. degrees in computer science and the Ph.D. degree in cryptography from Xidian University, Xi'an, China, in 2002, 2005, and 2008, respectively. He is currently a Professor with the School of Computer Science and Technology, Xidian University. He is also an Associate Director of Shaanxi Key Laboratory of Network and System Security and a member of the State Key Laboratory of Integrated Services Networks, Xidian University. His research interests include wireless network security and cloud computing security. He has also served on the technical program committees for several international conferences, including ICEBE, INCoS, CIS, and SOWN.
\end{IEEEbiography}

\vfill
\end{document}